\documentclass[a4paper,11pt]{article}

\usepackage{fullpage}
\usepackage[numbers]{natbib}
\usepackage{hyperref}
\hypersetup{colorlinks=true,linkcolor=blue,citecolor=blue,urlcolor=blue}

\usepackage{amsmath}
\usepackage[capitalize]{cleveref}

\usepackage{thm-restate}

\usepackage{amssymb}
\usepackage{xcolor}

\newcommand{\vals}{\textbf{v}}

\newcommand{\items}{{\cal{M}}}
\newcommand{\agents}{{\cal{N}}}

\newcommand{\FF}{{\cal{F}}}

\newcommand{\MMSi}[0]{MMS_{n}(v_i)}
\newcommand{\MMSv}[0]{MMS_{n}(v)}

\newcommand{\PSv}[0]{PS_{{n}}(v)}

\newcommand{\CC}{{\mathcal{C}}}
\newcommand{\QQ}{{\mathcal{Q}}}
\newcommand{\PTAS}[0]{PTAS(\epsilon)}

\newcommand{\Sharefullval}[1]{s(#1,n)}
\newcommand{\Sharefullpval}[1]{s'(#1,n)}
\newcommand{\Sharefullpv}[0]{\Sharefullpval{v}}
\newcommand{\Sharefullpvp}[0]{\Sharefullpval{v'}}
\newcommand{\Sharefullv}[0]{\Sharefullval{v}}
\newcommand{\Sharefullvi}[0]{\Sharefullval{v_i}}
\newcommand{\Sharefullvip}[0]{\Sharefullval{v'_i}}
\newcommand{\Sharefullvp}[0]{\Sharefullval{v'}}
\newcommand{\ShareGfullval}[1]{\hat{s}(#1,n)}
\newcommand{\ShareGfullv}[0]{\ShareGfullval{v}}
\newcommand{\ShareGfullvi}[0]{\ShareGfullval{v_i}}
\newcommand{\ShareGfullvp}[0]{\ShareGfullval{v'}}
\newcommand{\ShareGfullcv}[0]{\ShareGfullval{c\cdot v}}
\newcommand{\ShareGtwo}[2]{\hat{s}_{{#1}}({{#2}})}
\newcommand{\ShareG}[0]{\ShareGtwo{v}{v'}}
\newcommand{\ShareGfunc}[0]{\ShareGtwo{v}{\cdot}}
\newcommand{\ShareGvv}[0]{\ShareGtwo{v}{v}}
\newcommand{\ShareGvpvp}[0]{\ShareGtwo{v'}{v'}}
\newcommand{\ShareGcvcv}[0]{\ShareGtwo{c\cdot v}{c\cdot v}}
\newcommand{\ShareGvcv}[0]{\ShareGtwo{v}{c\cdot v}}
\newcommand{\PTASALG}[0]{PTAS_{ALG}}

\newcommand{\Prefix}[0]{Z}

\newtheorem{theorem}{Theorem}%[section]
\newtheorem{lemma}{Lemma}

\newtheorem{observation}[lemma]{Observation}
\newtheorem{proposition}[lemma]{Proposition}
\newtheorem{corollary}[theorem]{Corollary}% reset theorem numbering for each chapter
\newtheorem{question}{Question}
\newtheorem{definition}{Definition}%[section] % definition numbers are dependent on theorem numbers
\newtheorem{example}{Example}

\newtheorem{remark}[definition]{Remark}

\newenvironment{proof}{\noindent\bf{Proof.}\rm}{\hfill$\blacksquare$\bigskip}

\begin{document}

\title{Fair Shares: Feasibility, Domination and Incentives}
	\author{Moshe Babaioff\thanks{Microsoft Research ---  E-mail: \texttt{moshe@microsoft.com}.}, Uriel Feige\thanks{Weizmann Institute and Microsoft Research ---  E-mail: \texttt{uriel.feige@weizmann.ac.il}.}}

\maketitle

\begin{abstract}

We consider fair allocation of a set $\items$ of indivisible goods to $n$ equally-entitled agents, with no monetary transfers. Every agent $i$ has a valuation function $v_i$ from some given class of valuation functions. A \emph{share} $s$ is a function that maps a pair $(v_i,n)$ to a non-negative value, with the interpretation that if an allocation of $\items$ to $n$ agents fails to give agent $i$ a bundle of value at least equal to $s(v_i,n)$, this serves as evidence that the allocation is not fair towards $i$. For such an interpretation to make sense, we would like the share to be \emph{feasible}, meaning that for any valuations in the class, there is an allocation that gives every agent at least her share.  The maximin share (MMS) was a natural candidate for a feasible share for additive valuations. However, Kurokawa, Procaccia and Wang [2018] show that it is not feasible.

We initiate a systematic study of the family of feasible shares. We say that a share is \emph{self maximizing} if truth-telling maximizes the implied guarantee (the worse true value of any bundle that gives the share with respect to the report). We show that every feasible share is dominated by some self-maximizing and feasible share. We seek to identify those self-maximizing feasible shares that are polynomial time computable, and offer the highest share values. We show that a SM-dominating feasible share -- one that dominates every self-maximizing (SM) feasible share -- does not exist for additive valuations (and beyond). Consequently, we relax the domination property to that of domination up to a multiplicative factor of $\rho$ (called \emph{$\rho$-dominating}). {For additive valuations we present shares that are feasible, self-maximizing and polynomial-time computable. For $n$ agents we present such a share that is $\frac{2n}{3n-1}$-dominating, and is $\frac{4}{5}$-dominating when $n\leq 4$. For two agents we present such a share that is $(1 - \epsilon)$-dominating.  Moreover, for each of these shares we present a polynomial time algorithm that computes allocations that give every agent at least her share.}
\end{abstract}

\pagebreak

\section{Introduction}\label{sec:intro}
There is a large literature concerning the fair partition of indivisible goods 
among agents with equal entitlements to the items, without monetary transfers. Central to this problem is the definition of what is ``fair'' and when should an agent claim that she was treated unfairly be considered founded. One can roughly partition the approaches in the literature to two. The first aims to eliminate (or minimize) envy between the agents.\footnote{This approach is based on envy-freeness and its relaxations, we discuss it in Section \ref{sec:intro-related}.} The other aims to give each agent the value she ``deserves'', or close to it. 

In this paper we follow 
the latter approach, considering fairness based on the value that an agent obtains from the items she gets (relative to her own valuation).
Prior work have defined some specific notions of \emph{shares}, and in particular, the maximin share (known as MMS) which is the highest value that the worse-off agent gets, when all agents have the same valuation. 
Intuitively, a share aims to capture the value the agent should always get, no matter what the valuations of the others are.  
The maximin share (MMS) was a natural candidate for a share such that for any additive valuations {of the participating agents}, there is an allocation that gives every agent her share (as this is so when valuations are identical). 
However, Kurokawa, Procaccia and Wang 
\cite{KurokawaPW18} 
show that  there are instances for which it is not possible to give every agent her MMS share. 

We are not aware of prior work that attempts to formally define what a ``share'' means as a general concept, and which properties it should satisfy.
In this paper we aim to address this issue, formalizing the notion of a \emph{share} and of a \emph{feasible share}, studying relation between different shares (domination) and 
studying incentives for truthful valuation reporting.   
A share may serve as a fairness notion for a given class of valuations if it is \emph{feasible}, meaning that for any valuations of the agents in the class, there is an allocation that gives every agent at least her share. 
When a share is feasible, allocations that do not give some agent her share may plausibly be considered not to be fair. Thus, a feasible share serves as a lower bound on the value that an agent can expect to get, for any valuations of the others. As discussed above, the MMS is not feasible. 

\subsection{Shares as Contracts}
\label{sec:contracts}

We introduce properties that we would like shares to have. In order to motivate these properties, we take the view that shares can be interpreted as contracts. This view and its implications are presented in this section.

We consider fair allocation of a set $\items$ of indivisible goods to $n$ {equally-entitled} agents, with no monetary transfers. Every agent $i$ has a valuation function $v_i$ from some class of valuations $\CC$. The valuation  assigns a value $v_i(S)$ to any bundle $S\subseteq \items$. Valuations are assumed to be normalized ($v_i(\emptyset)=0$) and monotone non-decreasing.
We formalize the notion of a \emph{share} as a function that maps a pair $(v_i,n)$ to a non-negative value 
{$\Sharefullvi$, and the function satisfies  
some minimal requirements} (independence of item names\footnote{That is, permuting the names of the items does not change the value of the share. In particular, the share does not rely on different agents agreeing on the names of the items.} 
and being \emph{realizable}: 
{no larger than the value of the most valuable bundle, which is $v(\items)$ when items are goods).}
Note that to capture the ex-ante symmetry between the players {(all have equal entitlement to the goods)}, the share function is anonymous: the function is the same for all agents and does not depend on the {identity of the agent applying it.} 

It is instructive to think of the notion of a share as if it is a {\em contract} between the agent and the entity that produces the allocation.
We refer to this entity as the {\em allocator}. (Depending on the setting, the allocator can take one of many different forms, such as an actual person doing the allocation, or an algorithm doing the allocation.) The context of the contract is a set $\items$ of items, a number $n$ of agents, and a class $\CC$ of valuation functions. The terms of the contract are that the agent reports an allocation function $v$ in the class $\CC$, and the allocator guarantees {that the agent will get a bundle that she values at least as} 
(her {\em share}) $\Sharefullv$. Consequently, for a given $v\in \CC$, the contract partitions the set of bundles into two: those bundles $B$ that are \emph{acceptable}, 
namely, satisfy  $v(B) \ge \Sharefullv$, and those that are {not acceptable} ($v(B) < \Sharefullv$). 
{Given a share $s$ and valuation $v$, the \emph{share guarantee} $\ShareGfullv$ is {defined to be} the value of the least valuable acceptable bundle for $v$.}
The contract gives some rights to the agent, namely, the right to receive an acceptable bundle. The contract also gives rights to the allocator, namely, the right to choose which of 
the acceptable bundles to give to the agent. 

We now list properties that we desire contracts to have. They motivate concepts considered in the current paper.

\begin{itemize}

    \item Feasibility. The allocator offers shares to all agents of the allocation instance. Viewing shares as contracts, a natural requirement is that they do not contradict each other, so that the allocator can fulfil his side of the contract with respect to every agent. {In the terminology of the paper, we say that a share is {\em feasible} if for every allocation instance there is an allocation that gives every agent an acceptable bundle.} {Such an allocation will be referred to as an \emph{acceptable allocation}.}

    \item Legibility ({\em poly-time computability}). A contract should be written in such a language that the parties to the contract can understand its content. In the context of shares, the content is the value of the function $\Sharefullv$.
    In a computational setting, we interpret the notion of ``understanding'' as that of being able, given $n$ and $v$ (or if $v$ does not have a succinct representation, then access to $v$ using value queries), to compute in polynomial time the value $\Sharefullv$ of the share.  Thereafter, upon receiving a bundle $B$, the agent can determine whether the allocator has fulfilled his side of the contract. The agent can also determine whether she was fortunate to get a value much larger than her share. Likewise, the allocator can test whether his proposed allocation to the agents fulfills his side of the {contract with every agent}.

    \item Incentive {for truthful reporting} ({\em self maximizing}).  Suppose that an agent evaluates a contract by her value to the worst possible acceptable bundle under the contract.\footnote{The assumption that agents are risk averse is common in the literature of fair division,  e.g., when considering the behaviour of the splitting agent in the \emph{cut-and-choose} protocol (equal cut is optimal only under worse-case assumption about the other's valuation). We assume that this risk aversion extends beyond the valuations of others, and holds also to with respect to choice of allocation (satisfying the contract).} 
    Indeed, for most valuation classes of interest (e.g., unit demand, additive, submodular, subadditive), the agent might get this worst bundle, under an adversarial choice of an acceptable allocation and valuation functions for the other agents (see \cref{pro:SM}).
    As part of the contract, the agent reports a valuation function.
    Incentive {for truthful reporting} in this context means that the agent cannot get a better contract by reporting a different valuation function $v'$ instead of her true valuation function $v$. 
    This serves two purposes. One is that the agent need not waste efforts in figuring out which report gives her the best contract -- she can simply report her true valuation function. The other is that by receiving the true valuation functions of all agents, the allocator holds the relevant information that 
    is needed in order to serve {both the goal of giving each agent at least her share according to her true valuation function,  and} {in order to obtain any} auxiliary goals of the allocation mechanism (such as maximizing welfare), if there are such goals. The notion of a share  
    {giving incentive for truthful reporting}
    should not be confused with the notion of an allocation mechanism being incentive compatible. The former concerns getting the best contract, and the report of the agent can be given even before an allocation algorithm has been decided upon. The latter concerns getting the best bundle, and requires that the allocation algorithm be fixed before the agent reports a valuation function. We shall refer to shares that {give incentive for truthful reporting} 
    as being {\em self maximizing}, so as to avoid confusion with the notion of an allocation mechanism being incentive compatible.
\end{itemize}

An example of a share that is feasible, poly-time computable and self-maximizing, even when $\CC$ is the class of all monotone valuation functions {(items are goods)}, is a share that we shall denote by $s_n$, where {$s_n(v,n)$}  
is the value of the $n$-th most valuable item according to $v$. Feasibility is established by the following allocation mechanism, that we shall refer to as $AM_n$. Sort the agents in an arbitrary order, and let $v_i$ denote the valuation function reported by agent $i$. Items are allocated in rounds. In round $i < n$ agent $i$ gets an item of highest value according to $v_i$, among the items not yet allocated. In round $n$, agent $n$ gets all remaining items. 
Under $AM_n$, every agent $i$ gets a bundle that contains at least one of her top $n$ items,
and as valuation functions are monotone, agent $i$ receives a bundle of value at least $s_n(v_i,n)$. Poly-time computability of $s_n$ is evident\footnote{As $m$ value queries suffice in order to determine the value of every single item, and after sorting these values, the $n$-th largest value is the value of the share.}. The share $s_n$ is self-maximizing because for every report $v'_i$, the allocator is allowed to allocate to the agent any one of the top $n$ items according to $v'_i$, and the {true value of the worst of these items (value according to $v_i$)} cannot have value larger than $s_n(v_i,n)$, {the $n$-th most valuable item according to the true valuation $v_i$}. 

The reader may have observed that $AM_n$ is an incentive compatible allocation mechanism. Agent $i$ cannot improve her allocation by reporting a different $v'_i$ instead of her true $v_i$. Hence the share $s_n$ has an associated incentive compatible allocation mechanism. However, nothing in the definition of the share $s_n$ requires that this particular mechanism be used. For example, the allocator is free to implement a round robin allocation instead (agent $i$ selects items in rounds $i + kn$ for $k \in \{0, 1, 2, \ldots \}$, where in each such round she gets the yet not allocated item of highest marginal value relative to the set of items that she already holds). The round robin allocation gives every agent {value of at least her share value $\Sharefullv$,} 
but is not an incentive compatible allocation mechanism: depending on valuation functions of other agents, agent $i$ might improve her allocation by reporting a different $v'_i$ instead of her true $v_i$. 

One may further observe that {the allocation algorithm $AM_n$ actually}
offers every agent a share that {\em dominates} $s_n$ (always of value at least $s_n$, and sometimes higher). We refer to this share as $s_{n-1}$, where $s_{n-1}(v,n)$ is the minimum between the value of the $(n-1)$-th most valuable item according to $v$, and the value of the least valuable bundle among those containing $m - n + 1$ items. {Like $s_n$, the share $s_{n-1}$ is also feasible (by allocation algorithm $AM_n$), self maximizing ({as we prove later,} see \cref{prop:picking-well-behaved}, for example), and when valuations are additive, it is also poly-time computable.}

Given the above examples of $s_n$ and $s_{n-1}$, and the fact that agents are better off when $s_{n-1}$ is the share rather than $s_n$, it is natural to ask (for a given class $\CC$ of valuations) whether there is {a feasible share that is the ``best possible'' -- at least as large as any other feasible share.}
{Share  $s$ \emph{dominates} share $s'$ for a class of valuations functions $\CC$, if {for every $n$ (for which the share is defined) and} for every valuation $v\in \CC$ it holds that $\Sharefullv\geq \Sharefullpv$. For a given class $\CC$ of valuations, a feasible share is a \emph{dominant feasible share} if it dominates every other feasible share for that class. A feasible share is a \emph{SM-dominant feasible share} if it dominates every other feasible share that is self maximizing (SM).
}

\begin{question}
\label{Q1}
For which classes $\CC$ of valuations is there a share that is feasible (and better still, also self maximizing and poly-time computable) that dominates all other feasible shares (or at least dominates those feasible shares that are self maximizing)?
\end{question} 

As we shall see, there are classes $\CC$ for which SM-dominant feasible shares do not exist. 
{Moreover, we show that the class of additive valuations {over goods} does not have a feasible share that dominates every self-maximizing feasible share that is poly-time computable.}
This leads to a relaxed version of the above question, in which we look for a feasible share that dominates a scaled-down version of any other feasible share. 
{Formally, for $\rho \ge 0$, share  $s$ \emph{$\rho$-dominates} share $s'$ for a class of valuations functions $\CC$, if for every valuation $v\in \CC$ it holds that $\Sharefullv\geq \rho\cdot \Sharefullpv$. Feasible share  $s$ is a \emph{$\rho$-dominating share} for class $\CC$ if it $\rho$-dominates every other feasible share for $\CC$ {(self maximizing or not).}}

\begin{question}
\label{Q2}
For a given class $\CC$, what is the maximum value of $\rho$  (where $0 \le \rho \le 1$) for which there is a feasible share (and better still, also self maximizing and poly-time computable) that is $\rho$-dominating for $\CC$? 
\end{question}

{In \cref{app:auxiliary} we discuss auxiliary goals and the fact that there might be a tradeoff between such goals and maximizing $\rho$ when picking a $\rho$-dominating share. In \cref{sec:declaration} we discuss the possibility of combining a guarantee to give agents at least their share with additional non-binding ``statements of intent" to provide allocations with additional desirable properties.

Our paper is concerned with the above questions. We next describe our results. }

\subsection{Our Results} \label{sec:results}

We start by considering Question~\ref{Q1}, which asks for which classes of valuations there exists a share that is feasible and that dominates every other feasible share (or at least dominates those feasible shares that are self maximizing). We first observe that whenever the MMS is feasible, it is a feasible share that dominates every other feasible share, and the unique one (see \cref{prop:feasibleMMS-maximal}). For example, this is the case for the class of unit-demand valuations, in which the MMS (which is the value of the $n^{th}$ most valuable item) is feasible. In contrast, our main negative result shows that for any class of valuations that contains the additive valuations, a \emph{dominant feasible share} (a share that dominates every other feasible share) does not exist. Moreover, there does not exist a \emph{SM-dominant feasible share}: a share that dominates every other feasible share that is self maximizing (SM).

\begin{restatable}{theorem}{thmMainImpossibility}
\label{thm:main}
    Let $\CC$ be any class of valuations that contains all additive valuations over goods. 
    For $\CC$, there is no feasible share that dominates all feasible shares that are self maximizing. 
    Thus, there does not exist an SM-dominant feasible share for class $\CC$. 
\end{restatable}

Note that the impossibility result holds although we do not require the dominating share itself to be self maximizing. 
Additionally, note that our  impossibility result does not rely on bounded computational power -- it holds even when there are no computational constraints. 

We prove the theorem by defining a family of shares that are feasible and are self maximizing for any class $\CC$ of valuations.\footnote{For additive valuations these shares are also poly-time computable.} 
We also show that for each additive valuation, there is a share in the family that has value that is the same as the MMS for that valuation. Thus, to dominate all these shares, the dominating share must dominate the MMS. 
To complete the proof we use the fact that the MMS is not feasible for additive valuations (\cite{KurokawaPW18}).

The theorem can also be proven as a corollary of another result that we prove and is of independent interest. 
We show that any feasible share can be transformed to a feasible share that dominates it, and is self-maximizing. 
Thus, in a sense, we can always assume that a feasible share is self-maximizing. 
This transformation is not done in polynomial time, thus, when computation is important, 
this transformation does not show that self maximization comes for free (computationally). 

\begin{restatable}{theorem}{thmFeasibleToSMFeasible}
   \label{thm:ICdominates}
    Let $\CC$ be an arbitrary class of valuation functions, and let $s$ be an arbitrary feasible share for class  
    $\CC$. 
    Then for class 
    $\CC$ there exists a feasible share $s'$ 
    that dominates $s$ and is  
    self-maximizing.
\end{restatable}

\cref{thm:main} shows that there is no SM-dominating feasible share for additive valuations and beyond.  
The direct implication of this result is that there is no ''ultimate''  feasible share for these classes: no matter how we define a feasible share, there will always be room for additional shares not dominated by it. 

We thus move our focus to Question \ref{Q2}, and consider the problem of finding a share that is poly-time computable and self maximizing and is a 
$\rho$-dominating feasible share. 
That is, we relax exact domination to approximate domination, and look for shares that are $\rho$-dominating and have $\rho$ that is as large as possible. We first {observe} that for any class of valuations $\CC$, a share is a $\rho$-dominating feasible share if and only if it $\rho$-dominates the MMS (see \cref{prop:dom-MMS-all}).

{The rest of our paper focuses on additive valuations over items that are goods}.
Observe that as the MMS is not feasible for more than two agents, and hard to compute even for two agents, we cannot hope to find a $\rho$-dominating share for $\rho=1$,  and not even for two agents (if we want the share to be polynomial time computable). We thus look for 
poly-time computable and self maximizing feasible shares that are  
$\rho$-dominating, with $\rho<1$ that is as large as possible. 

We present a class of shares for additive valuations  {over goods} 
which we call {\em ordinal maximin shares}, and show that every share in this class is self maximizing. {Moreover, for two agents, any share in this class is feasible.}
We then present {several $\rho$-domination results, 
for each we define} a share that is an ordinal maximin share (thus self maximizing), {and we show that each is feasible} and can be computed in polynomial time. 

\begin{theorem} \label{thm:intro-add}
For the class of additive valuations  {over goods} 
with $n$ agents, there exists a share that is feasible, self-maximizing, polynomial-time computable and $\frac{2n}{3n-1}$-dominating. 
{For {at most four agents this share is} $\frac{4}{5}$-dominating.} 
{Furthermore, for two agents,
{for any positive $\epsilon$} there exists a share that is feasible, self-maximizing, polynomial-time computable and is $(1 - \epsilon)$-dominating. }
\end{theorem}

{Interestingly, the result for four agents implies the first known poly-time allocation algorithm that gives every agent a $\frac{4}{5}$ fraction of her MMS (the prior result of $\frac{4}{5}$ {by \cite{GM19}} was only existential).} 

We remark that it is known~\cite{GT20} that for additive valuations it is possible to find an allocation that gives every agent $\left(\frac{3}{4} + \frac{1}{12n}\right)$  {fraction} of her MMS.  
This immediately implies that the $\left(\frac{3}{4} + \frac{1}{12n}\right)$-MMS is a share that is feasible. 
Yet, this share is not self maximizing, as so is the case for any $\rho$-MMS share for $\rho<1$ (see \cref{prop:shares-self}). 
One can apply \cref{thm:ICdominates} on this share, to transform it to a feasible share that is self maximizing and $\left(\frac{3}{4} + \frac{1}{12n}\right)$-dominating, but the transformation is not poly-time computable. 
Alternatively, one can consider the poly-time algorithm of~\cite{GT20} that finds an allocation that gives every agent $\frac{3}{4}$ of her MMS and use it to define a {polynomial time computable} share as follows: for valuation $v$ the share is defined to be the lowest value of any of the $n$ bundles that the algorithm outputs, when the input is $n$ agents each with valuation $v$ (same valuation $v$ for all agents). However, such a share would not be self maximizing\footnote{For example, if $n=2$ and item values are $9,5,5,1$, then the bundle containing only the first item is {acceptable}
under this share. However, if the values $6,5,5,4$ are reported instead, then every {acceptable} bundle {that might be allocated by the algorithm of~\cite{GT20}} has value at least~10, with respect to both the reported valuation and the true valuation.}. 
Moreover, it is not clear whether this share is feasible\footnote{One might hope that given additive valuations $v_1, \ldots, v_n$, running the allocation algorithm of~\cite{GT20} on these valuation functions will produce an allocation that gives every agent at least her share. However, this might not be true, because the share value for each agent $i$ is determined by running the algorithm on a different instance, one in which all valuation functions are $v_i$. A-priori, there is no reason to believe that in the former run agent $i$ would necessarily get a higher value than the share value determined by the latter run.}.
We leave open the problem of finding the largest value of $\rho$ for which there is a $\rho$-dominating share that is feasible, self maximizing, and poly-time computable {for additive valuations. }

\subsection{Related Work}\label{sec:intro-related}

In our work we initiate a systematic study of the notion of fair shares in item allocation problems. We are not aware of any previous such systematic studies. However, various shares have been proposed in previous work, and were extensively studied. For divisible items and additive valuations, the proportional share (PS) can always be guaranteed, 
or using the terminology of the current paper, it is feasible. For indivisible items, the PS is not feasible. This served as partial motivation for introducing other types of shares, and the one which has been most extensively studied is the maximin share (MMS), introduced in~\cite{Budish11}. For additive valuations, the MMS is no larger than the PS, though for non-additive valuations it may be larger.
It turns out that for additive valuation and $n \ge 3$, also the MMS is not feasible~\cite{KurokawaPW18}. 
This led to multiple works [\citealp{KurokawaPW18}, \citealp{amanatidis2017approximation}, \citealp{GM19}, \citealp{BK20}, \citealp{GhodsiHSSY18}, \citealp{garg2019approximating} ,\citealp{GT20}, \citealp{FST21}] trying to determine the best ratio $\rho$ such that in instances with additive valuations, there always is an allocation giving every agent a bundle of value at least a $\rho$ fraction of 
her MMS (a $\rho$-MMS allocation). The current known bounds on $\rho$ for general $n$ are $\frac{3}{4} + \Theta(\frac{1}{n}) \le \rho \le 1 - \frac{1}{n^4}$. For $n=3$, the bounds are $\frac{8}{9} \le \rho \le \frac{39}{40}$.  The definition of MMS easily extends to non-additive valuations, though the values of $\rho$ for which $\rho$-MMS allocations are guaranteed to exist deteriorate the more we extend the class of valuations (see \citealp{GhodsiHSSY18}). 

The fact that the MMS is not feasible led to consideration of other notions of shares. A closely related notion is the ``$1$-out-of-$d$'' share (see~\cite{HosseiniSearns21}, for example), in which in an instance with $n$ agents, the share of the agent is the value of her MMS in instances with $d$ agents. If $d$ is sufficiently larger than $n$ ($d = 2n$ certainly suffices, and it is an open question whether $d=n+1$ suffices), then the ``$1$-out-of-$d$'' share becomes feasible. However, it does not provide any approximation to the MMS ({for $d>n$} the ``$1$-out-of-$d$'' share might be~0 even when the MMS is strictly positive). A different type of share is the Truncated-Proportional Share (TPS)~\cite{BEF2021c}. Its value (for additive valuations) lies between the PS and the MMS, and hence it is not feasible. (The TPS was not introduced so as to handle feasibility issues, but for different reasons: using the TPS simplifies proofs concerning the approximation ratios achievable by allocations.) 

{The main focus of our paper is the fair}
allocation of indivisible \emph{goods} 
to agents with \emph{equal} entitlement. Shares have been extended to other related settings. One such setting is allocation of chores (items of negative value). See for example~\cite{AzizRSW17, huang2019algorithmic}. Another such setting is that of agents of arbitrary entitlement. This setting led to definitions of additional types of shares, including the WMMS~\cite{farhadi2019fair} (which in fact does not qualify as a share under our definitions, as it depends on more than just the valuation and entitlement of the agent), the ``$\ell$-out-of-$d$'' share~\cite{BabaioffNT2020}, and the AnyPrice Share (APS) 
~\cite{BEF2021b}. 
We {briefly discuss} these related settings in Section~\ref{sec:discussion}. 

In the current paper, the criterion for ``fairness'' of allocations is based on shares. Other fairness criteria have also been extensively studied. For additive valuations, the notion of {\em envy freeness} 
provides a strengthening of the proportionality requirement, in the sense that every envy free allocation (one in which no agent envies the bundle received by a different agent) offers every agent at least her proportional share, and hence at least also her MMS. (For non-additive valuations, envy free allocations need not give every agent her MMS.) Envy free allocations exist for divisible items (and additive valuations), but need not exist for indivisible items. A relaxation of envy freeness, envy-free-up-to-one-good (EF1)~\cite{Budish11}, is feasible for every class of valuation functions~\cite{LiptonMMS04}, but does not even offer {better than a $\frac{1}{n}$} approximation to the MMS when valuations are additive\footnote{{For example, in an EF1 allocation in which an agent gets a single item of value~1, and every other agent gets one item of value~1 and one item of value~$n$.}}. 
An intermediate notion, envy-free-up-to-any-good (EFX)~\citep{CKMPSW19}, does offer a constant approximation to the MMS when valuations are additive, but it is not known if it is feasible {beyond two agent {for general valuations}, and beyond three agents when valuations are additive} (this is a well known open question~\cite{Procaccia20}). In the current paper, we do not address envy-based fairness notions.

{
A self-maximizing share provides incentives to risk-averse agents to report their true valuations, in settings in which they do not know the inner details of the allocation mechanism and the valuations of other agents. However, if agents do know (or have beliefs about) what valuation functions others agents report and how the allocation algorithm actually works, then the question of whether the underlying share is self-maximizing might become irrelevant. This is because the agent no longer needs to use the share value as an indication to what value she will receive -- she can instead simulate runs of the allocation algorithm {to rank her possible reports}
based on her beliefs, and report accordingly.
The allocation algorithm $AM_n$ presented in Section~\ref{sec:contracts} (and well known also from previous work~\cite{ABM2016}) is incentive compatible even in such full information settings, but does not offer the agents a good approximation to their MMS value. More generally, no allocation algorithm that offers a good approximation to the MMS (and allocates all items) is incentive compatible in the full information setting~\cite{ABM2016}. {In contrast, we are able to combine incentives and a good approximation to the MMS, but our incentives are not for an \emph{allocation algorithm} but rather for the \emph{share guarantee}.
}  
}

\section{Model} 
\label{sec:model}

We consider fair allocation of a set  $\items$ of $m$ indivisible  items to a set $\agents$ of $n$ agents with equal entitlements to the items, without monetary transfers. 
Each agent's preference over subsets of items is captured by a \emph{valuation} $v$ which is a function that maps every subset of items $S\subseteq \items$ to a non-negative value $v(S)$ that the agent derives from getting $S$.
We assume that $v$ is normalized: $v(\emptyset) = 0$. 
If $v$ is non-decreasing then the items are referred to as {\em goods} (with respect to $v$). 
Throughout the paper we assume items are goods, except when explicitly stated otherwise. 

A valuation $v$ {over the set of goods $\items$} is called \textit{additive} if each {good}  
$j\in \items$ is associated with a value $v(j)\geq 0$, and the value of a set of goods $S \subseteq \items$ is $v(S)=\sum_{j\in S}v(j)$.
A  valuation $v$ is called \textit{unit-demand} if each {good} 
$j\in \items$ is associated with a value $v(j)\geq 0$, and the value of a set of {goods} 
$S \subseteq \items$ is $v(S)=\max_{j\in S} v(j)$.

{Let $v_i(\cdot)$ denote the valuation of agent $i$, and let $\vals =(v_1,\ldots,v_n)$ denote the vector of agents' valuations.}
We assume that there are no transfers (no money involved). 
The set of items $\items$ needs to be partitioned between the agents in a fair way.  
In this work we consider deterministic allocations (partitions). An \emph{allocation} $A$ is a partition of the items to 
$n$ disjoint bundles  $A_1, \ldots, A_n$ (some of which might be empty), where $A_i \subseteq {\items}$ for every $i\in \agents$.
{We denote the set of all allocations of $\items$ to the $n$ agents by $\mathcal{A}$.} 

We next present the {definitions of several types of shares that were introduced in prior work.} 
For a setting with $n$ agents with valuations over a set $\items$ of items, the {\em proportional share (PS)} of {an agent with valuation $v$  
is $\PSv = \frac{v(\items)}{n}$.}
For indivisible items, it might not possible to give every agent her proportional share, for example, if there are fewer items than agents. The \emph{maximin share (MMS)} aims to address {this issue. 
The MMS of an agent is the highest value that she} can ensure herself by splitting the goods to $n$ bundles and getting the worst one.

\begin{definition}[maximin Share (MMS)]
	\label{def:pes}
	The \emph{maximin share (MMS)} of {an agent with valuation $v$} 
	over a set of items $\items$, when there are $n$ agents, which we denote  by $\MMSv$,  
	is defined to be the highest  value she can ensure herself by splitting the item in $\items$ to $n$ bundles and getting the worst one. 
	Formally: 
	$$\MMSv =   \max_{(A_1,\ldots,A_n)\in \mathcal{A}} \min_{j \in [n] } \left\{v(A_j) 
	\right\},$$
\end{definition}
When {$v$ and  $n$ 
are clear from context we simply use the term MMS to denote the MMS share $\MMSv$.}

Somewhat surprisingly, there are allocation instances with additive valuations over goods for which there is no allocation that gives every agent her MMS~\cite{KurokawaPW18}. 
{In this work we study shares can be concurrently given to all agents, the relations between such shares, and properties of such shares.}

\section{Shares and the Quest for a Dominant Feasible Share}\label{sec:maximal}
{In this section we formally define the notions that we study in this paper and prove some basic properties of these notions. We start with the definition of a \emph{share} and the  
\emph{share guarantee}.
We then move to discuss \emph{feasible shares} and \emph{self-maximizing shares}.  
We conclude with \emph{share domination} and the notions of \emph{dominant feasible shares} and \emph{SM-dominant feasible shares}. }

{
We start by formalizing the notion of a \emph{share} for a class of valuations. 
A \emph{share} for a class of valuations $\CC$ is a function mapping the set $\items$ of items, a valuation $v$ in $\CC$} and number $n$ of agents, to a real value. 
The share is required {to be {\em realizable} (no larger than the value of the most valuable bundle, which is $v(\items)$ when items are goods)} and to not 
depend on names of items. To give an agent her share, we must give her a bundle of value at least her share - the minimal such value is defined to be her \emph{share guarantee}. 

\begin{definition}
A {{\em share function} (or simply \emph{share})} $s$ {for a class of valuations $\CC$} is a function that outputs a real value {$\Sharefullv$} when given two parameters: 
a valuation function {$v\in \CC$ over a set of items $\items$,} and the number of agents $n$. 
It needs to satisfy two properties:
\begin{itemize}
    \item Realizability: for every $v\in \CC$ and $n$ it holds that $\Sharefullv \le \max_{S\subseteq \items}v(S)$.\footnote{For the case that items are goods this condition is simply $\Sharefullv \le v(\items)$. {We present this general definition so that it will apply to valuations that are not necessarily over goods, e.g, apply when items that are chores (bads). }} 
    \item Name independence: renaming the items does not change the value of the share.
\end{itemize}
{Given a share function $s$, a bundle $S$ is \emph{acceptable} for agent with valuation $v$ and share $s$, if its value is at least the share guarantee: $v(S) \ge \Sharefullv$.}

Given a share function $s$, the associated {\em share guarantee} $\hat{s}$ is defined as the function satisfying 
$$\ShareGfullv \doteq \min_{\{S \; | \; v(S) \ge \Sharefullv\}} v(S).$$ 
Namely, for any given input,  the share guarantee is the minimum value of
 a bundle whose value is at least as high as the share {for that input}.
Alternatively, for any valuation $v$, the share guarantee is the lowest value of any acceptable bundle for the valuation $v$ and share $s$.
\end{definition}

Observe that for every share $s$, the realizability property implies that the corresponding share guarantee is well defined. {When items are goods, the realizability property implies that the corresponding share guarantee is non-negative}. Every allocation that gives an agent a bundle of value at least equal to the share, necessarily gives the agent a bundle of value at least equal to the share guarantee. As a simple example, if 
there are two agents ($n=2$) and 
there are two items with $v(e_1) = 1$ and $v(e_2) = 2$, then the proportional share guarantee is $2$ (whereas the proportional share is only $\frac{3}{2}$). 

{We remark that for any share that is defined to be a value of some bundle (like the MMS), the share guarantee is the same as the share itself. Yet, for a share that is not defined as a value of a bundle (like the proportional share), the share guarantee might be strictly larger than the share.}

The intended interpretation of a share guarantee is as a guarantee to an agent $i$ who reports her true valuation function $v_i$: if the underlying allocation algorithm is one that gives every agent at least her share (according to the reported valuation functions), then regardless of the reports of all other agents, agent $i$ is guaranteed to receive a bundle of value at least equal to her share guarantee. {Thus, we next discuss the concept of a \emph{feasible share}, a share for which it is always possible to give every agent an acceptable bundle according to the share.
}

\subsection{Feasible Shares}\label{sec:notion-feasible}

{We focus on shares for which every agent can always get her share.}
For a class $\CC$ of valuation functions, we say that a share $s$ for $\CC$ is \emph{feasible} if for any agents' valuation vector $\vals =(v_1,\ldots,v_n)$ such that $v_i\in \CC$ for every $i$, there {exists an \emph{acceptable} allocation $A$: 
an allocation that gives every agent an acceptable bundle. That is, for every agent $i$ the bundle $A_i$ has value at least her share: $v_i(A_i)\geq \Sharefullvi$.}
{Note that by definition, an allocation that is acceptable must also gives every agent $i$ a bundle of value at least her share guarantee, that is, $v_i(A_i)\geq \ShareGfullvi$. }

It is well known that feasible shares have the advantage that even in strategic environments, agents will get their shares under equilibrium behaviour.\footnote{Such properties are much harder to prove for envy-based fairness notions like EF1, see {\cite{AmanatidisBFLLR21}}.} 
Fix any class of valuations $\CC$, a share $s$ that is feasible for $\CC$, and any algorithm for share $s$: an algorithm that for any vector $\vals' =(v'_1,\ldots,v'_n)$ of $n$ reported valuations in $\CC$ outputs an allocation $A'$ that gives each agent her share with respect to her reported valuation ($v'_i(A'_i)\geq \Sharefullvip$).  The mechanism that is defined by the algorithm is not necessarily (dominant strategy) truthful {(as is the case for the round-robin mechanism for the share $s_n$ presented in Section~\ref{sec:contracts}).} 
Nevertheless, as the share is feasible, each agent has a ``safe strategy'' (being truthful) which guarantees that she gets some minimal utility (her share guarantee). 
This property implies that outcomes of the mechanism will have nice utility guarantees even when agents are acting strategically:  if the algorithm is deterministic, in any pure Nash equilibrium of the game, each agent's
utility must be at least her share guarantee. Similarly, if agents are risk neutral, if the algorithm is randomized then in any mixed Nash equilibrium of the game each agent's expected  utility must be at least her share guarantee.\footnote{Similar claim is also true for any mixed Nash equilibrium in Bayesian settings.}

As observed above, for any deterministic algorithm for a feasible share, in any pure Nash equilibrium of the game, each agent's utility must be at least her share guarantee.
Yet, not every allocation algorithm will induce a game in which a pure Nash equilibrium exists. Does there always exist some algorithm for a feasible share $s$ for which a pure equilibrium always exists? We next show that such an algorithm indeed exists. 
\begin{proposition}\label{prop:nash-exists}
   Let $s$ be any feasible share. There exists an algorithm for share $s$ that induces a full-information game in which a pure Nash equilibrium always exists, and the equilibrium allocation is acceptable. 
\end{proposition}
\begin{proof}
The game is defined as follows. First an arbitrary order over the agents is fixed. Without loss of generality we assume that agents' names are according to the given order (so agent $1$ is the first, $2$ is second, etc.). 
In the first stage, agents are approached according to that order, and each is requested to report a valuation, after being informed of the reports of all prior agents. We denote the reported valuation of agent $i$ by $v'_i$. 
Let $S'$ be the family of all acceptable allocations for share $s$ with respect to the reported valuations $(v'_1,v'_2,\ldots,v'_n)$. As the share is feasible the set $S'$ is non-empty. 
Now, in the second stage, a deterministic algorithm (that was fixed in advance) picks an allocation out of $S'$ and implements it. 
For example, the algorithm might pick the first (according to some fixed order) welfare maximizing allocation in $S'$.

We next claim that in this game there exists a pure Nash equilibrium (and even a subgame-perfect equilibrium). This holds as this is a full-information sequential game, and thus can be solved by backward induction. 
Specifically, in this subgame-perfect equilibrium, when considering her report, agent $n$, knowing all reports $(v'_1,v'_2,\ldots,v'_{n-1})$ of other agents and  the deterministic algorithm that will be used to pick an allocation, reports some valuation $v'_n$ that maximizes her utility. This serves as the base case for the induction. 
Similarly, any agent $i$ (from $n-1$ to $1$), knowing all reports $(v'_1,v'_2,\ldots,v'_{i-1})$ of other agents before her, and the strategies used by the agents that report after her (as defined recursively by this procedure), reports some valuation $v'_i$ that maximizes her utility. These strategies clearly define a subgame-perfect equilibrium. The allocation is acceptable as any agent can guarantee she gets her share by being truthful, so her report must give her at least as high utility as her share. 
\end{proof}

We note that the algorithm might not run in polynomial time, and even if it does, it might be computationally hard to find these equilibrium strategies. 

\subsection{Self-Maximizing Shares}\label{sec:notion-self-max}

We next define the notion of a {\em self-maximizing} share.
We  say that a share function is {\em self-maximizing} {for a class of valuations $\CC$ if for every valuation function in $\CC$}, reporting the true valuation function is a dominant strategy for maximizing 
the lowest value bundle (with respect to the true valuation) among all bundles that are acceptable for the reported valuation.
{To formalize this, we first define some useful notation. 
{For share $s$ and any valuation {$v\in \CC$}, define the share \emph{implied guarantee} 
when reporting $v'\in\CC$:}
$$\ShareG = \min_{\{S \; | \; v'(S) \ge \Sharefullvp\}} v(S)$$

\noindent to be the lowest {true value of any set (measured by the true valuation $v$),} among all sets that 
are acceptable with respect to $v'$.
Observe that {the implied guarantee when the report is the true valuation is simply the share guarantee, that is, } $\ShareGvv = \ShareGfullv$.

{Note that in the definition of the implied guarantee, $v'$ 
only influences the {family of bundles that are considered acceptable, yet the values of these bundles are}
measured with respect to the \emph{true} valuation $v$.
}

{The implied guarantee $\ShareG$ depends on $\Sharefullvp$ which depends on} $n$, but as  $n$ will always be clear from the context we do not make this dependence explicit in the notation, in order to keep the notation simple.

To illustrate the concept of the share implied guarantee, consider the proportional share $s$ over additive valuations. 
Assume that $n=3$ and consider valuations over four items, with items of values $(5,4,4,2)$ for valuation $v$. The value of the share is $\Sharefullv = 5$ and the share guarantee is $\ShareGfullv = 5$ as well. For reported valuation $v'$ with reported item values $(4,4,4,3)$, it holds that $\Sharefullvp = 5$ so no single item forms an acceptable bundle, and the set of acceptable bundles for $v'$ contains all bundles of size at least 2. The implied guarantee is the true value of the least valuable pair, which has value of $6$ and thus the implied guarantee when reporting $v'$ is $\ShareG =6$. Observe that this is larger than the share guarantee ($\ShareG = 6> \ShareGfullv = 5$), so by misreporting the implied guarantee increases.       

{With the notation of implied guarantee} 
we can now formally define the notion of self-maximizing share, {a share in which truthful reporting ($v'=v$) maximizes $\ShareGfunc$, the implied guarantee}:}

\begin{definition}
{  
A share $s$ is {\em self-maximizing} for class of valuations $\CC$
if for every set $\items$ of items and number of agents $n$, for {every true valuation $v\in\CC$ and every report $v'\in\CC$} 
it holds that 
$$\ShareGvv = \ShareGfullv\geq \ShareG$$ 
}
\end{definition}

Let us briefly discuss the relation between a share being self-maximizing and incentives to agents to report their true valuation functions (``truthfulness''). Suppose first that we are in a setting in which the agent knows the items and her valuation function. 
The agent does not know the valuation functions of other agents, and does not know the details of the allocation algorithm, except for the fact that it gives every agent at least her share guarantee with respect to its input {(the share is feasible)}, for a share 
that is self-maximizing. In such a setting, for a risk averse agent who wishes to maximize the minimum possible value that it receives (minimum possible with respect to the information that the agent does not have), being truthful is a dominant strategy {(for any valuation class for which the worse case can indeed be realized, as is the case for most classes of interest, see Appendix \ref{sec:only-guaranteed}).}
By being truthful she guarantees to herself at least her share guarantee, and no misreporting of her valuation function offers a higher guarantee. However, for an agent that is either not risk averse (for example -- in a Bayesian setting in which there is some probability distribution over valuations of other agents and over the choice of allocation algorithm to use), or knows the details of the allocation algorithm, truthfulness need not be a dominant strategy. 

\subsubsection{The MMS is Self Maximizing but any $\rho-$MMS is not}
\label{sec:MMS-self}

{We next show that the MMS is self maximizing {for any class of valuations, yet even for additive valuations,} for any positive $\rho<1$, the $\rho$-MMS is not. Additionally, we observe that the proportional share is not self maximizing {for $n \ge 3$ even for additive valuations (while it is so for $n=2$)}.}
{For the proof see \cref{app:prop-shares-self}.}

\begin{restatable}{proposition}{propSelfMaxShares}
\label{prop:shares-self} 
    The MMS is a self-maximizing share {for every class $\CC$ of valuations}. 
    In contrast, even just for additive valuation functions, 
    the  $\rho$-MMS (for any $0 < \rho < 1$) is not self-maximizing for $n \ge 2$, and 
    the proportional share is not self-maximizing for $n \ge 3$ {(while it is so for $n=2$)}.
\end{restatable}

\subsubsection{Properties of Self-Maximizing Shares}\label{sec:prop-self}

In this section we show that any self-maximizing share {for class of valuations $\CC$} satisfies several natural properties: 

\begin{itemize}

\item {\em Scale Invariant:} 
{Assume that class of valuations $\CC$ is closed under multiplication by any positive constant ($v\in \CC$ implies $v'=c\cdot v$ is in $\CC$ for any $c>0$).}
A share function $s$ is \emph{scale invariant} {for $\CC$ if for every $c > 0$, scaling the valuation function $v\in \CC$} by a multiplicative factor of $c$, also scales the share guarantee by the same multiplicative factor $c$: 
{$\ShareGfullcv= c\cdot \ShareGfullv$ for every $c>0, v\in \CC$.}

    \item {\em Monotone:} A share function $s$ is \emph{monotone} {for $\CC$ if the share guarantee is a monotone function of the valuation function:} 
    {for every $v,v'\in \CC$ such that $v \ge v'$ (meaning that $v(S) \ge v'(S)$ for every $S \subseteq \items$) it holds that 
    $\ShareGfullv\geq \ShareGfullvp$.
    }
    
    \item {\em $1$-Lipschitz:} A share function $s$ is  {\em $1$-Lipschitz } {for $\CC$} if the following holds for every $\delta>0$ and valuation function $v\in \CC$:
    if the value of {every} 
    set changes by at most $\delta$ then the share guarantee increases by at most $\delta$. In other words,
    {for valuations $v,v'\in \CC$, if $|v'(S)- v(S)|\leq \delta$ for every set $S \in \items$,  then 
    $\ShareGfullvp\leq \ShareGfullv +\delta$.
    }
\end{itemize}

\begin{restatable}{proposition}{propSelfMaxProperties}
\label{prop:self-max-properties}
    Any share that is self maximizing {for $\CC$,} is monotone, $1$-Lipschitz and scale invariant {for $\CC$}. 
\end{restatable}

{We present the proof in  \cref{app:prop-self-max-properties}.}

\subsection{The Dominant Feasible Share}\label{sec:dominant-feasible}

{Recall that a share is feasible if for any agents' valuation vector there exists an acceptable allocation. Clearly the share function that is identically zero is feasible, yet not very attractive as we want the feasible share to be ``as large as possible''. We next formalize what this means.}
One can naturally define a partial order over share functions. 
A share (function) $s$ \emph{dominates} share (function) $s'$ for a class of valuations functions $\CC$, if for every {$n$ (for which $s$ and $s'$ are defined) and every} valuation $v\in \CC$ it holds that $\Sharefullv\geq \Sharefullpv$. For a given class $\CC$ of valuations, is there a feasible share that is ``the best'' in the sense that it dominates all other feasible shares for this class? 
We say that a feasible share $s$ is a \emph{dominant feasible share} for a class of valuations functions $\CC$, if for the class $\CC$ {the share $s$} dominates any other feasible share $s'$. 
It is easy to see that any instance that demonstrates that the MMS is not feasible can be used to show that a dominant feasible share does not exist: Consider for example the instance $I$  presented in \cite{KurokawaPW18} for three agents with valuations denoted by $(v_1,v_2,v_3)$, for which it is not possible to give every agent her MMS. We can define three feasible shares $s_i$ for $i\in \{1,2,3\}$ in which $s_i(v_i,3)$ equals to the MMS of $v_i$, and for any $v\neq v_i$ we have $s_i(v,3)=0$. Each such share is clearly feasible. For instance $I$, any  share $s$ that dominates all these three shares must give every agent her MMS, but that is {known to be} impossible.  

The proof above feels very unsatisfactory as the three shares that are defined in the proof are clearly very artificial. Expecting a share to dominate even those artificial feasible shares seems too much to ask for. 
We thus 
only ask that $s$ dominate the subset of feasible shares that are also self maximizing. As being self maximizing is a strong requirement, it seems much more plausible that there might exist a feasible share that dominates only those  feasible shares that are self-maximizing.  
We say that a feasible share $s$ is a \emph{SM-dominant\footnote{SM stands for ''Self-Maximizing''. 
}
feasible share} for a class of valuations functions $\CC$, if for the class $\CC$ it dominates every other feasible and self-maximizing {(SM)} share $s'$.

Our main negative result is that for additive valuations (and beyond), even if we look for a feasible share that dominates only those {feasible} shares that are \emph{self maximizing}, such a share does not exist. {Note that even for SM-domination, while the dominating share $s$ needs to be feasible, we do \emph{not} require it to be self maximizing (making our negative result  stronger).} Before discussing the proof we discuss the {relation between the MMS and} 
the concept of SM-dominant feasible shares.

We first observe that for valuation classes for which the MMS is feasible, the MMS is the unique SM-dominant feasible share ({see proof in  \cref{app:prop-feasibleMMS-maximal}}):

\begin{restatable}{proposition}{propMMSSMdominating}
\label{prop:feasibleMMS-maximal} 
	Consider a setting with $n$ agents and a set of items $\items$.
{For any class of valuations $\CC$, the MMS 
	dominates every feasible share for $\CC$.
	{Moreover, if the MMS  
	is feasible, then it is the unique dominant feasible share for $\CC$,
	and the unique SM-dominant feasible share for $\CC$.}  }
\end{restatable}

Two immediate corollaries follow. First, for the class of unit-demand valuations, for which the MMS is the value of the $n$-th most valuable item, the MMS is feasible (and self maximizing), so it is the unique SM-dominant feasible share for these valuations. 
Second, for two agents and additive valuations the MMS is feasible (obtained by \emph{cut-and-choose}) so it is the unique SM-dominant feasible share for these valuations. 
Yet, beyond two agent the MMS is known to be infeasible for additive valuations \cite{KurokawaPW18} so the proposition {cannot be used to find an SM-dominant feasible share.} 
Yet, it does not rule out that maybe there is a different SM-dominant feasible share for additive valuations and more than two agents.  

In the next section we consider additive valuations as well as more general classes of valuations that contain additive valuations (e.g., submodular and sub-additive) and show that in contrast to the case of unit-demand valuations, SM-dominant feasible shares for those classes do not exist. 
This {nonexistence result} holds even if we do not require the dominating share to be self maximizing.   
{We prove this claim}
by presenting an explicit family of feasible shares that are self maximizing and have the following property: for any additive valuation there is a share in the family for which the share for this valuation equals to the MMS for this valuation. The result is then derived from the fact that  the MMS for additive valuations is not feasible.     

We remark that our {nonexistence result}
can also be derived as a corollary (Corollary \ref{cor:feasible-SM}) from our {result that shows} that any feasible share is dominated by a feasible share that is also self maximizing (\cref{thm:ICdominates}). Yet, for the class of additive valuations, while the shares defined in the first proof are poly-time computable, our transformation from a feasible share to a {self maximizing} feasible share that dominates it {(\cref{thm:ICdominates})} uses super-polynomial computation.

\section{No Dominant Feasible Share for Additive and Beyond}\label{sec:no-dominant}

Our first main result shows that for any class {of valuations} that contains additive valuations, there is no feasible share that dominates even only those feasible shares that are self-maximizing. Thus, there is no SM-dominant feasible share for {any of these}
classes of valuations. 

\thmMainImpossibility*

We note that this result shows that there is no SM-dominant feasible share, even if we do not restrict that dominating share itself to be self maximizing. 
We also highlight the fact  that the non-existence of an SM-dominant feasible  share does not rely on computational feasibility -- it holds even with unbounded computational power. 

The proof of the theorem uses shares that are based on picking orders {(also refered to as picking sequences)} that we define below. Thus, before presenting the proof of the theorem (Section \ref{sec:proof-main}), we present these shares (Section \ref{sec:picking}). 

\subsection{The $\omega$-Picking-Order Share}
\label{sec:picking}

{To prove \cref{thm:main} we need to show that for any class $\CC$ of valuation functions that contains all additive valuations, there is no feasible share that dominates every feasible self-maximizing share. To do so, for every such class $\CC$, we design a collection $\QQ$ of shares, where every share in $\QQ$ is feasible and self maximizing. The feasibility and self maximization properties need to hold for {all} valuations in $\CC$, not just for additive valuations. To prove that there is no SM-dominant feasible share we show that no feasible share $s$ dominates every share in $\QQ$: 
for every share $s$ that is feasible for $\CC$ (even if $s$ is not self maximizing), there is some share $s'$ in $\QQ$
that for some valuation function $v\in \CC$ offers higher value than $s$ does.}

The class $\QQ$ that we consider is 
based on picking orders, that is, sets that are created by the $n$ agents picking items according to some predefined order in which they pick. Such shares, mostly based on the round robin picking sequence, have been introduced also in previous work (see~\cite{CFS17, GLW21}), {though our treatment is concerned with arbitrary picking sequences.}
Such shares will be defined even for non-additive valuations, and will be feasible and {self maximizing} (although not necessarily poly-time computable {when the valuations are not additive}).

For a setting with $n$ agents and a set $\items$ or size $m$, a \emph{picking order $\omega$} is a list $\omega_1,\omega_2,\ldots, \omega_m$, with $\omega_j\in[n]$ for each $j\in [m]$. 
{The value $\omega_j$ determines which agent picks an item at the $j$-th step in the picking sequence.}

Fix a valuation $v$ and a picking order $\omega$. 
Let $x_k$ be the maximum value that agent $k\in [n]$ with valuation $v$ can ensure she gets when the picking order is $\omega$. That is, the maximum value when facing an adversary that makes all picks  when $\omega_j\neq k$ in order to minimize the value that agent $k$ gets. Note that 
computing the strategy of {the agent} and of such an adversary is easy for additive valuations (by greedily picking items from high to low value according to the agent's valuation), but not necessarily so for 
{more general classes of valuations.} 

Fix a picking order $\omega$.
The \emph{$\omega$-Picking-Order share} of an agent with valuation $v$ is $\min_{k\in [n]} x_k$. That is, the minimum over all possible identities of the agent, of the maximum she can ensures herself with that identity when the picking order is $\omega$. 
{Note that as these shares are defined to be a value of some bundle, for each picking order $\omega$ the share guarantee is the same as the share itself.}

To illustrate this share, let us consider the case that the valuation is additive. In this case, {both the agent and the adversary} will simply pick items {in order of decreasing value according to the agent's valuation (breaking ties arbitrarily)}.
Note that tie breaking (by the agent or the adversary) in picking has no effect on the final value the agent gets.
In this case, the share will simply be the minimal value that any of the $n$ agents get, when every agent has valuation $v$ (all valuations are the same), and simply pick items in non-increasing order of values. 
Consequently, when agents are additive, the picking order share is {a special case of a more general family of shares that we define in Section \ref{sec:additive-ordinal}, which we call \emph{ordinal maximin shares} (see Definition~\ref{def:ordinal}).
We show that all shares in this {general} family are  self maximizing (\cref{prop:ordinal}), and hence {picking order shares are self maximizing when valuations are additive.} }
{Moreover, we show that picking-order shares are self maximizing for all valuations (even non-additive).}

{We summarize our results for picking shares in the next proposition (see proof in  \cref{app:prop-picking}).}

\begin{proposition}\label{prop:picking}
Fix any class $\CC$ of valuations. 
For any picking order $\omega$, the $\omega$-Picking-Order share is feasible and self maximizing for $\CC$.  Moreover, if valuations are additive then an acceptable allocation can be computed in polynomial time. 
{Finally, for any additive valuation $v$ there exists a picking order $\omega_v$ for which the $\omega_v$-Picking-Order share for an agent with valuation $v$ equals to the MMS of that agent, that is, to $\MMSv$. }
\end{proposition}

\subsection{Proof of \cref{thm:main}}\label{sec:proof-main}

In this section we restate and prove \cref{thm:main}, our main negative result.

\thmMainImpossibility*

\begin{proof}
Fix any share $s$. 
By {\cref{prop:picking},}
for every picking order $\omega$ and any valuation (including all valuations in the class $\CC$), 
the $\omega$-Picking-Order share is self maximizing and feasible.
Thus, if $s$ dominates every self-maximizing  and feasible share it must dominate every $\omega$-Picking-Order share. 

As $\CC$ includes all additive valuations, $s$ must, in particular, dominate every self-maximizing and feasible share for additive valuations.
As for additive valuations, by {\cref{prop:picking},} 
for some picking order $\omega$ the $\omega$-Picking-Order share is the same as the MMS, we conclude that for additive valuations, for $f$ to dominate every self-maximizing and feasible share, it must dominate the MMS.   

Yet, 
{for agents with additive valuations, there are instances for which no allocation gives each agent her MMS~\cite{KurokawaPW18}}. We thus conclude that for class $\CC$, the share $s$ does not dominate every self-maximizing and feasible share.
\end{proof}

Our {nonexistence}
result shows in particular that an SM-dominant feasible share for additive valuations does not exist. Given this negative result, in Section \ref{sec:additive} we consider approximate domination, and present several positive results for additive valuations.  
Before focusing on additive valuations we present a general transformation from any feasible share to a dominating share that is feasible and self maximizing. {This result also implies an alternative proof of our nonexistence result.}

\section{Domination Relations among Shares}\label{sec:domination-ALL}
\subsection{Domination by a Self-Maximizing Share}\label{sec:domination-SM}

In this section we show that any feasible share can be transformed to a feasible share that dominates it, and is self-maximizing. Thus, in a sense, we can always assume that a feasible share is self-maximizing. 
Yet, this transformation is not done in polynomial time, thus, when computation is important, this proof does not show that  self maximization comes for free. 

\thmFeasibleToSMFeasible*

{The main idea in the proof of \cref{thm:ICdominates} is to construct $s'$ as follows. For every valuation $v$, the value of $s'(v)$ is equal to the implied guarantee of $s$ with respect to the report $v'$ that maximizes this implied guarantee. The details of how this idea is turned into a proof
are presented in \cref{app:thm:ICdominates}. }

{
From \cref{thm:ICdominates} we derive the following immediate corollary.
\begin{corollary}
    \label{cor:feasible-SM}
    Let $\CC$ be an arbitrary class of valuation functions, and let $s$ be an arbitrary feasible share for class 
    $\CC$. The share $s$ is a dominant feasible share for $\CC$ if and only if
    $s$ is a SM-dominant feasible share for $\CC$.
\end{corollary}
The corollary combined with the simple observation that there is no dominant feasible share for additive valuations (implied by the fact the the MMS is not feasible, see Section \ref{sec:dominant-feasible}),  gives an alternative proof to \cref{thm:main} for any class that contains the additive valuations.}

We have shown that any feasible share can be transformed to a feasible share that dominates it, and is self-maximizing. 
Yet, this transformation is not done in polynomial time {({not} even for the class of additive valuations)}. 
This suggests the following open problem: find such a transformation that can be computed in polynomial time when valuations are additive and the original share is polynomial time computable, or prove that the problem is computationally hard. 

\subsection{$\rho$-Domination}\label{sec:domination-MMS}

An immediate corollary of Proposition \ref{prop:feasibleMMS-maximal} is that for any class of valuations $\CC$, a share $s$ that $\rho$-dominates the MMS also $\rho$-dominates every feasible share. This is stated in the next proposition, together with two stronger versions of the converse claim.  
\begin{proposition}\label{prop:dom-MMS-all} 
    For any class of valuations $\CC$, a share $s$ that $\rho$-dominates the MMS also $\rho$-dominates every feasible share.
    
    For any class of valuations $\CC$, a share $s$ that $\rho$-dominates every 
	feasible share that is self maximizing also $\rho$-dominates the MMS. 

	Furthermore, for the class of additive valuations, a share $s$ that $\rho$-dominates every 
	feasible share that is self maximizing and poly-time computable, also $\rho$-dominates the MMS. 
\end{proposition}
\begin{proof}
By \cref{prop:feasibleMMS-maximal}, for any class of valuations $\CC$ the MMS dominates every feasible share for $\CC$. An immediate corollary is that a share $s$ that $\rho$-dominates the MMS also $\rho$-dominates every feasible share for $\CC$.

Next, we prove that for any class of valuations $\CC$, a share $s$ that $\rho$-dominates every feasible share that is self maximizing also $\rho$-dominates the MMS. Indeed, 
for any valuation $v\in \CC$, let $s_v$ the share that is defined as follows: it equals the MMS on $v$, and is zero otherwise. This share is clearly feasible for $\CC$, by the MMS definition.
By \cref{thm:ICdominates}, the share $s_v$ is dominated by some share $s'_v$ that is feasible and self maximizing. The share $s'_v$ must also be equal to the MMS on $v$, as it is at least as large as $s_v$ on $v$, and is no larger since it is feasible.  
Thus, in order for $s$ to $\rho$-dominate all  self-maximizing feasible shares, it must dominate  $s'_v$ for every $v\in\CC$, and thus it must $\rho$-dominate the MMS for $\CC$.

Finally, we show that for the class of additive valuations, a share $s$ that $\rho$-dominates every feasible share that is self maximizing and poly-time computable, also $\rho$-dominates the MMS. 
By \cref{prop:picking}, 
for any additive valuation $v$ there exists a picking order $\omega_v$ for which the $\omega_v$-Picking-Order share for an agent with valuation $v$ equals to the MMS of $v$.
Additionally, by \cref{prop:picking}, for any additive valuation $v$ the $\omega_v$-Picking-Order share is feasible, self maximizing and can be computed in polynomial time. Thus, in order for share $s$ to $\rho$-dominate every feasible share that is self-maximizing and poly-time computable, it must $\rho$-dominate the $\omega_v$-Picking-Order share for every additive valuation $v$, and thus it must $\rho$-dominates the MMS.
\end{proof}

Denote by $\alpha$-MMS the share that is defined to be an $\alpha$-fraction of the MMS (that is $\alpha\cdot MMS$ for positive $\alpha\leq 1$). 
By definition, for every valuation the $\alpha$-MMS $\alpha$-dominates the MMS. Thus, for any class of valuations $\CC$, by \cref{prop:dom-MMS-all}, the $\alpha$-MMS $\alpha$-dominates every feasible share for $\CC$.

Fix the number of agents $n\geq 2$.
For a class $\CC$ of valuations, let $\alpha_{n}^*(\CC)$ be the supremum\footnote{
For the class of additive valuations the maximum is actually obtained, although  there are infinitely many possibilities for the number of items, and infinitely many possible valuation functions.  {See~\cite{FN22} for details.}
} 
over all values of $\alpha$ for which the $\alpha$-MMS is feasible for class $\CC$ for $n$ agents (for any number of items $m$). 
For $n$ agents, the $\alpha_{n}^*(\CC)$-MMS share $\alpha_{n}^*(\CC)$-dominates every feasible share for $\CC$ (by \cref{prop:dom-MMS-all}). Yet, by \cref{prop:shares-self}, this share is not self maximizing, {not even} for the class of additive valuations. 
One can use \cref{thm:ICdominates} to turn the $\alpha_{n}^*(\CC)$-MMS to a feasible self-maximizing share that dominates it, thus $\alpha_{n}^*(\CC)$-dominating every feasible self-maximizing share for $\CC$. Yet, the transformation of \cref{thm:ICdominates} is not polynomial time, {not} even for the class of additive valuations.
Thus, in the next section we focus on the class of additive valuations over goods and we seek a self-maximizing feasible share that is poly-time computable and $\rho$-dominates the MMS, for  $\rho$ that is as large as possible.

\section{Results for Additive Valuations over Goods}\label{sec:additive}
{In this section we focus on the class of additive valuations over goods (which we will simply call \emph{additive valuations} below, for short).} {Given the discussion at the end of the prior section (\cref{sec:domination-MMS}),} we seek self-maximizing feasible shares that are poly-time computable and {$\rho$-dominate the MMS, for $\rho$ that is as large as possible.}  

We first present a family of shares, referred to as \emph{ordinal maximin shares}, and show that all members of this family are self-maximizing  (see \cref{sec:additive-ordinal}). We then present several ordinal maximin shares, all of them are feasible, self-maximizing and polynomial-time computable.
For general $n$, we present such a share that is $\frac{2n}{3n-1}$-dominating. 
Furthermore, {for {up to four} agents we show that this share is $\frac{4}{5}$-dominating.
For two agents
and any positive $\epsilon$ we present such a share 
that is $(1 - \epsilon)$-dominating (a PTAS). (As the MMS is NP-hard to compute even for $n=2$, we cannot hope the get the MMS exactly by a poly-time share.) Moreover, for each of those shares we present a polynomial time algorithm that computes  allocations that give every agent at least her share. 
{The PTAS for two agents appears in \cref{sec:additive-two}.}
The other shares are presented in \cref{sec:additive-n}.
}

\subsection{Ordinal Maximin Shares}\label{sec:additive-ordinal}

For additive valuation functions, we define the class of {\em ordinal maximin shares}, {and we show that} 
every share in this class is self maximizing. 
{We slightly abuse the notion of share and instead of thinking about one share that is defined for every $n$, we think of shares that are each defined for a specific value of $n$ (e.g., this enables us to define a PTAS-based share only for the case that $n=2$).}
We first define the notion of a family of partitions.

\begin{definition}
{Fix $n\geq 2$. For any integer $m$,}
an $(m,n)$-partition $P$ of the set $[m]$ into $n$ parts is a collection of $n$ disjoint bundles that together contain the set: $P = (B_1, \ldots, B_n)$ satisfying $\cup_j B_j=[m]$, $B_k\cap B_j=\emptyset$ for every pair  $k\neq j$. (Some of the bundles might be empty.) For given $n$ and $m$, a family $\FF_{m,n}$ of partitions is collection of $(m,n)$-partitions. 
{For $n\geq 2$, a \emph{family $\FF_n$ of partitions} is a collection of families $\FF_{m,n}$ of partitions, one such family for every $m$.}
\end{definition}

\begin{definition}\label{def:ordinal}
{Fix the number of agents $n\geq 2$. 
A share for $n$ agents with} additive valuation functions is {\em ordinal maximin} if for every value of $m$ (number of items)
it has an associated family of partitions $\FF_{m,n}$, and for every additive valuation function $v$ over $m$ items, the corresponding share value can be computed as follows. 
\begin{enumerate}
    \item Sort items from highest value to lowest value according to $v$ (breaking ties arbitrarily -- the share value is not affected by how ties are broken). 
    Without loss of generality, we denote the $j$-th highest value item according to this order by $e_j$ (thus, $v(e_j)\geq v(e_k)$ for $j<k$). 
    \item Relative to this order, {for a set $S\subseteq [m]$, let $v(S)=\sum_{j\in S} v(e_j)$.} 
    For every partition $P = (B_1, \ldots, B_n) \in \FF_{m,n}$, let
    $v(P) = \min_{j}[v(B_j)]$. 
    \item The value of the share for $v$ is $\max_{P \in \FF_{m,n}}[v(P)]$. {A partition $P$ is \emph{optimal} if $v(P)$ equals to the value of the share ($P$ belongs to the argmax).}  
\end{enumerate}
Hence, an ordinal maximin share {for $n$ agents} is completely characterized by its associated family {$\FF_n$ of collections of partitions $\FF_{m,n}$ (one for every $m$).
For any such family {$\FF_n$,} we denote the associated ordinal maximin share over the
additive valuations 
by \emph{{$\FF_n$}-maximin share}.}
\end{definition}

The MMS is an example of a share that is ordinal maximin  {for every $n$ (with $\FF_{m,n}$ being the family of all $(m,n)$-partitions)}. The following {\em round robin share}, based on the {\em round robin} allocation procedure, is also ordinal maximin. For every $(m,n)$, the family $\FF_{m,n}$ contains a single partition $P = \{B_1, \ldots, B_n\}$, where for every $j$, bundle $B_j$ contain the items $e_{j+kn}$ (for non-negative integers $k$, where items are sorted from highest to lowest value). The value of the round robin share is equal to the value of $B_n$. More generally, every {\em picking order} induces an ordinal maximin share, in which $\FF_{m,n}$ contains only a single partition. Given valuation function $v$, bundle $B_i$ in this partition contains those items whose index is those rounds in which agent $i$ gets to pick in the picking order.

{We observe that for two agents, for  any family $\FF_2$ the $\FF_2$-maximin share is feasible {(see \cref{obs:ordinal-feasible} in \cref{app:additive-ordinal})}. 
For more than two agents, the $\FF_n$-maximin share might not be feasible
(as in the case for the MMS with $n=3$), but specific  families of partitions yield feasible shares (as we prove in Section \ref{sec:additive-n}).}

{The main attractive property of ordinal maximin shares is that they are all self maximizing  
(see proof in \cref{app:prop-ordinal}).}
\begin{restatable}{proposition}{lemmaOrdinalSelfMaximizing}
\label{prop:ordinal}
For any family $\FF_n$ of partitions, its associated ordinal maximin share $\FF_n$-maximin is self maximizing.
\end{restatable}

\subsection{A Self-maximizing PTAS for Two Agents}\label{sec:additive-two}
{For two agents with additive valuations, the MMS is a self-maximizing share that is feasible. However, computing the MMS is NP-hard {even for two agents, and thus we cannot hope the get the MMS exactly by a poly-time share (assuming $P\neq NP$)}. In this section, for the case of two agents with additive valuations and for every $\epsilon > 0$, we present a  share that we shall refer to as the $\PTAS$ share. For every additive valuation function, the value of the $\PTAS$ share is at least $(1 - \epsilon)$ times the value of the MMS.
Like the MMS, the $\PTAS$ share is self maximizing and feasible. However, unlike the MMS, for every $\epsilon > 0$, the $\PTAS$ share can be computed in polynomial time. Moreover, given the valuation functions of the two agents, an allocation giving each agent at least her $\PTAS$ share can be computed in polynomial time.

The $\PTAS$ share is based on the following principle. Recall that for $n=2$, the value of the MMS is a solution to the following {\em partition} problem. The input to the problem is a set $\items = \{e_1, \ldots, e_m\}$ of items and their values $\{x_1 = v(e_1), \ldots, x_m = v(e_m)\}$. The desired output is a partition of $\items$ into two disjoint bundles, $B_1$ and $B_2$, such that $\min[v(B_1), v(B_2)]$ {is maximized}, or equivalently,   $\min_{i\in \{1,2\}}\sum_{j \; | \; e_j \in B_i} x_j$, is maximized. As the partition problem is (weakly) NP-hard, so is the problem of computing the MMS. 

To overcome the NP-hardness issue, we use a polynomial time approximation algorithm, {referred to here as $\PTASALG$, for the partition problem. On every input instance, $PTAS_{ALG}$ returns a partition in which the value of the less valuable part of the partition is within a factor of $(1 - \epsilon)$ the value of the less valuable part in the optimal (most balanced) partition. We associate a share with $PTAS_{ALG}$, where for every additive valuation function, the value of the share is the value returned by $PTAS_{ALG}$. 
The $PTAS_{ALG}$ share is feasible, because on every allocation instance (for two agents), we can use the cut and choose procedure: consider the partition that $PTAS_{ALG}$ induces for agent~1, let agent~2 choose which part she takes, and agent~1 gets the remaining part. The $PTAS_{ALG}$ share is also polynomial time computable and approximates the MMS within a ratio of $(1 - \epsilon)$, by its definition. Hence it remains to address the issue of being self maximizing.} 

The partition problem has very efficient approximation algorithms: it has a fully polynomial time approximation scheme (FPTAS), that has polynomial dependence on $\frac{1}{\epsilon}$. However, known approximation algorithms for the partition problem to not appear to induce a share that is self maximizing. Instead, we design a new polynomial time approximation scheme (PTAS), with exponential dependence of $\frac{1}{\epsilon}$, which produces a share that is an ordinal maximin share, and hence self maximizing (by \cref{prop:ordinal}). 

With the above introduction in mind, we present the $\PTAS$ share in the proof of the following theorem. 

\begin{theorem}\label{thm:two-agents} 
	Consider a setting with two agents that have additive valuations over a set  of items $\items$.
	For any fixed $\epsilon>0$, the $\PTAS$ share (described in the proof of the theorem) is self-maximizing, computable in polynomial time, 
	and is at least a $1-\epsilon$ fraction of the MMS. 
    Additionally, the share is feasible, and {an acceptable allocation}  can be found in polynomial time.
\end{theorem}

\begin{proof}
{We may assume that $m \ge \frac{3}{\epsilon}$, as for $m < \frac{3}{\epsilon}$ we can set $\PTAS$ to be equal to the MMS. (The MMS can be computed in time exponential in $2^{O(1/\epsilon)}$ in this case, which is constant time when $\epsilon > 0$ is constant.)

We may also assume that $\MMSi \ge \frac{v_i(\items)}{3}$.
The only way that this assumption can be violated is if there is an item $e_1$ of value  $v_i(e_1) > \frac{2 v_i(\items)}{3}$, but then again we can take $\PTAS$ to be equal to the MMS (whose value in this case is $v_i(\items) - v_i(e_1)$).}

Given a desired approximation ratio of $1 - \epsilon$, let $k = \left \lceil \frac{3}{2\epsilon} \right \rceil$. 
{Note that $k \leq m$ as $m \ge \frac{3}{\epsilon}\geq \left \lceil \frac{3}{2\epsilon} \right \rceil = k$.}
Partition $\items$ into the set $\Prefix$ 
of the $k$ most valuable items (breaking ties arbitrarily), and the set $S$ of remaining items.  Order the items of $S$ in order of decreasing value (breaking ties arbitrarily),  
and let $e_m$ denote the last item in this order.
Consider the following family {$\FF_{m,2}$} of at most $2^k m$ partitions of $\items$. A partition of $\items$ into two bundles $(B_1,B_2)$ is a member of {$\FF_{m,2}$} if and only if $B_1 \cap S$ is a suffix of $S$. That is, bundle $B_1$ contains an arbitrary subset of $\Prefix$, but the set of items that $B_1$ contains from $S$ must be a (possibly empty) set of consecutive items, ending in $e_m$.  The value of the PTAS share for agent $i$ is $\PTAS_i = \max_{(B_1,B_2)\in {\FF_{m,2}}}\min[v_i(B_1),v_i(B_2)]$. That is, the definition of $\PTAS_i$ is similar to $\MMSi$, except that instead of maximizing (the minimum value of a bundle) over all possible partitions, one maximizes only over the partitions in ${\FF_{m,2}}$. For every item $e$ in $S$ 
we have that $v_i(e) \le \frac{v_i(\Prefix)+v_i(e)}{k+1}  \le \frac{v_i(\items)}{k+1}$. We claim that this implies that $\PTAS_i \ge \MMSi - \frac{v_i(\items)}{2(k+1)}$.
Recalling the assumption that $\MMSi \ge \frac{v_i(\items)}{3}$ and the choice of $k = \left \lceil \frac{3}{2\epsilon} \right \rceil$, the claim implies that $(1 - \epsilon)\cdot \MMSi \le \PTAS_i \le \MMSi$.

We now prove the claim {that $\PTAS_i \ge \MMSi - \frac{v_i(\items)}{2(k+1)}$.}
Let $(B^*_1, B^*_2)$ denote the optimal partition. Construct a partition $(B_1,B_2)\in {\FF_{m,2}}$ as follows. For items in $\Prefix$, we choose $B_i \cap \Prefix = B^*_i \cap \Prefix$ (for $i \in \{1,2\}$). If either $B^*_1$ or $B^*_2$ contain no items from $S$, then the partition  $(B^*_1, B^*_2)$ itself is one of the partitions in ${\FF_{m,2}}$.  Hence we may assume that both $B^*_1$ and $B^*_2$ contain items from $S$. Consequently, {$\max[v_i(B^*_1 \cap \Prefix),v_i(B^*_2 \cap \Prefix)] \le \frac{v_i(\items)}{2}$}. By choosing for $B_1$ either the longest suffix of $S$ for which $v_i(B_1) \le \frac{v_i(\items)}{2}$ or one additional item from $S$, we get that $\min[v_i(B_1), v_i(B_2)] \ge \frac{v_i(\items)}{2}-\frac{1}{2}\max_{i\in S}[v_i(e)] \ge \MMSi - \frac{v_i(\items)}{2(k+1)}$, as claimed.

{
$\PTAS_i$ can be computed by the following algorithm that uses only addition operations and comparison operations. Sort all items in non-increasing order of values. This can be done using $O(m \log m)$ comparison operations. The set $\Prefix$ is composed of the first $k$ items, and $S$ is composed of the remaining items. 
For every suffix of $S$, compute the sum of item values in the suffix. This can be done using $O(m)$ addition operations (starting from the last suffix, and extending the suffix one item at a time). 
Thereafter, there are $2^k$ possible partitions of $\Prefix$. For each such partition $(\Prefix_1,\Prefix_2)$, the sum of item values in $\Prefix_1$ can be computed using $O(k)$ addition operations. The optimal choice of suffix of $S$ to add to $\Prefix_1$ can be found using $O(\log m)$ operations, by binary search over the suffixes. Hence $\PTAS_i$ can be computed using $O(m\log m + 2^{k} (k + \log m))$ operations. As $k = O(\frac{1}{\epsilon})$, and the time to perform a single operation is at most linear in the input size (note that the number of input bits required in order to represent the value of an item might be much larger than $m$, and be the dominating factor in the input size), this running time is polynomial in the input size, for every fixed $\epsilon > 0$ (and also when $\frac{1}{\epsilon}$ is logarithmic in the input size).
}

{As $\PTAS$ is an ordinal maximin share, it is feasible for two agents by Observation \ref{obs:ordinal-feasible}, and self maximizing by \cref{prop:ordinal}.} 
\end{proof}

The approximation algorithm used in the proof of \cref{thm:two-agents} in order to compute the $\PTAS$ share can be extended to any number $n$ of agents. This requires, among other changes, to increase $k$ to $n \lceil \frac{1}{\epsilon} \rceil$ {in order for the approximation to still hold}. As the running time is exponential in $k$, it becomes exponential in $n$, but is still polynomial in $m$ for every fixed $n$. However, for $n \ge 3$ and $\epsilon \le \frac{1}{4}$, the resulting share is not feasible (unlike the case for $n=2$), {as it is the same as the MMS for instances in which it is impossible to give every agent her MMS. 
Indeed, recall that for} every $n \ge 3$,  there are allocation instances with additive valuations and $m \le 4n$ items in which no allocation gives every agent her MMS~\cite{KurokawaPW18}. In such instances and when $\epsilon \le \frac{1}{4}$, the value of $k$ is {at least $4n$, and as $4n \ge m$ the family} ${\FF_{m,n}}$ contains all possible partitions {of $m$ items to $n$ agents}. Consequently, in these instances, for every agent, $\PTAS_i = \MMSi$. Hence no allocation gives every agent her $\PTAS$ share.}

\subsection{The Class of Nested Shares}\label{sec:additive-n}

We introduce a class of ordinal maximin shares for additive valuations that we refer to as $NS$, standing for \emph{Nested Shares}. For $n$ agents, a share in this class {is parameterized by a parameter $q\in [n]$.  
It will be convenient to have both $n$ and $q$ explicit in the notation of the share, so we use $NS_{n,q}$ to denote this share for parameters $n$ and $q$.}

For $n$ agents, given parameter $q$, we next define the \emph{Nested Share} $NS_{n,q}$. 
Sort the items in decreasing value (breaking ties arbitrarily), and denote the $j$-th highest value item by $e_j$. Like the MMS, if the number of items $m$ satisfies $m < n$ then $NS_{n,q}(v) = 0$, and if $m = n$, then $NS_{n,q}(v) = v(e_n)$. Hence we present the definition for $m > n$. 
A \emph{consecutive partition} of a set of items $Z$ is defined by $n$ indexes $i_1\leq i_2\leq\ldots \leq i_n$ (we also denote $i_0=0$). These indexes define a partition of the items to $n$ bundles of consecutive items: for $j\in[n]$  bundle $Z_j$ contains items $\{e_{i_{j-1} +1}, \ldots, e_{i_j}\}$ ($Z_j$ is empty if $i_j=i_{j-1}$). 

The family of partitions that define the share $NS_{n,q}$ is defined as follows. Fix an integer variable $k$, whose value satisfies {$n-q < k \le m$.} 
For $k$,
consider the following family ${\FF_{m,n,q,k}}$ of {(at most $m^{n + q -2}$)}
partitions $(B_1, \ldots, B_n)$ of $\items$ that are defined as follows. 
Let $Z$ be the set of items $\{e_1,\ldots, e_k\}$. 
Fix any consecutive partition $(Z_1,\ldots,Z_n)$ of $Z$ to $n$ 
sets, such that $i_j=j$ for every $j \le n-q$ (that is, the $n-q$ first sets are singletons). 
Then, fix any consecutive partition {$(S'_1,\ldots, S'_n)$ of suffix $S'= S= \items\setminus Z$ to $n$ sets (possible empty), and for convenience denote $S_j=S'_{n-j+1}$ for $j\in [n]$.}
A partition of $\items$ into $n$ bundles $(B_1, \ldots, B_n)$ is a member of ${\FF_{m,n,q,k}}$ if and only if for some pair of partitions as defined above,
{$B_j= Z_j\cup S_j$}
for every $j\in[n]$. 

The value of the $NS_{n,q}$ share for valuation $v$ {over $m$ items} is defined to be \\ $NS_{n,q}(v) = \max_k \max_{(B_1,\ldots, B_n)\in {\FF_{m,n,q,k}}}\min_j[v_i(B_j)]$.
The definition of the {$NS_{n,q}$ share} is similar to that of the MMS, except that instead of maximizing (the minimum value of a bundle) over all possible partitions, 
one maximizes only over the partitions in $\FF_{n,q} = {\bigcup_{n-q < k \le m} \FF_{m,n,q,k}}$.\footnote{{Note that $\FF_{n,q}$ depends also on $m$. Yet, the only way in which we use $\FF_{n,q}$ is in the definition of  the share $NS_{n,q}$, which accepts a valuation as an argument. As 
$m$ could always be deduced from the valuation, to simplify notation we omit explicitly including $m$ in the notation.}
}
{We remark that $NS_{n,q}$ shares are called the \emph{Nested Shares} as for every $j\in[n-1]$ all items in $B_{j+1}$ are nested inside the items of $B_{j}$ (within the sorted order).}

Observe that for every $n$ and every $q > q'$ {(and every $m$)} it holds that $\FF_{n,q'} \subseteq \FF_{n,q}$, and consequently $NS_{n,q}(v) \ge NS_{n,q'}(v)$ for every valuation $v$.  In particular, $NS_{n,n}$ dominates $NS_{n,q}$ for every $q\in [n]$.

\begin{example}\label{example:NS}
To illustrate the $NS_{n,q}$ share consider the following example with four agents ($n=4$). The valuation is additive and item values are  (3, 3, 2, 2, 2, 2, 2, 2, 1, 1). Here, the MMS is~5 (the MMS partition contains two bundles of the form $(3,2)$ and two bundles of the form $(2,2,1)$). However, for every integer $q$ such that $1 \le q \le 4$, the value of the nested share $NS_{4,q}$ is~4. {To lower bound $NS_{4,q}$, it suffices to consider $q=1$.} The $NS_{4,1}$ partition is required to place the first four items in separate bins, and may reach a value of~4 by placing at least one additional item in each bin (e.g., creating bins $(3,1)$, $(3,1)$, $(2,2,2)$, $(2,2,2)$).
{To upper bound $NS_{4,q}$, it suffices to consider $q=4$.} The $NS_{4,4}$ partition allows for $Z_j$ to be of any size, but still it must create at least one bin of value at most~4. ({There are four items of odd value. For all bins to have value~5, the odd valued items must be in different bins. However, in every $NS_{4,4}$ partition, either there is a bin with only a single item, or two odd valued items are in the same bin.}) 
\end{example}

Though some of our results apply to $NS_{n,q}$ shares for arbitrary values of $n \ge 2$ and $1\le q \le n$, our most interesting current results apply to the case of $q=3$.

\begin{theorem}\label{thm:NS2} 
	Consider a setting with $n$ agents that have additive valuations over a set  $\items$  of goods.
    Then the {$NS_{n,3}$ share} is self-maximizing, poly-time computable, and is at least a $\frac{2n}{3n-1}$ fraction of the MMS, and at least a $\frac{4}{5}$ fraction of the MMS when {$n\leq 4$}.
    Additionally, the share is feasible, and an acceptable allocation can be found in polynomial time.
\end{theorem}

We outline the main ideas of the proof of \cref{thm:NS2}, and refer the reader to \cref{app:thm-NS2} for further details. 
The fact that $NS_{n,3}$ is self maximizing is an immediate consequence of the fact that it is an ordinal maximin share (\cref{prop:ordinal}). Polynomial time computability of the $NS_{n,q}$ share value (for every $1 \le q \le n$) and an associated optimal partition is demonstrated by an algorithm based on dynamic programming. See Lemma~\ref{lem:NS-poly-time}. The $\frac{2n}{3n-1}$ approximation holds already for $q=1$, and the proof of this fact can be derived from results in~\cite{GMT19}, showing that a related allocation algorithm provides allocations that give every agent a $\frac{2n}{3n-1}$ fraction of her MMS. In our case, we only need to prove this approximation ratio when all agents have the same valuation function, making the proof simpler. See Lemma~\ref{lem:twothirdsapprox}. 
We also complement this result by showing that even for $q=3$, as $n$ grows, the approximation ratio of $NS_{n,3}$ converges (from above) to $\frac{2}{3}$. See Proposition~\ref{pro:NEWnegative}. {The $\frac{4}{5}$ approximation for $n=4$ uses in an essential way the choice of $q=3$, and does not hold for $q \le 2$. (Likewise, the $\frac{4}{5}$ approximation for $n=3$ requires $q \ge 2$, and does not hold for $q = 1$.) The $\frac{4}{5}$ approximation {crucially depends on considering} 
partitions in which the first $n$ items are not all placed in separate bins. 
See Lemma~\ref{lem:rho32} (for $n=3$) and \cref{lem:rho43-ub} (for $n=4$). We remark that the proof of \cref{lem:rho43-ub} is computer assisted, see Appendix \ref{sec:computer}.

\subsubsection{Proof that $NS_{n,3}$ is a feasible share}

Here we prove part of  \cref{thm:NS2}, showing that the $NS_{n,q}$ share is feasible for $q = 3$ {(thus also for $q=2$ and $q=1$).}  Moreover, we show that an acceptable allocation can be computed in polynomial time.
We do not know whether the $NS_{n,q}$ share is feasible 
for larger values of $q$.}

{For our feasibility proof, we shall make use of the notion of \emph{ordered allocation instances}.

\begin{definition}
We say that an allocation instance is {\em ordered} if for every {$j$ and $k$ such that} $1 \le j < k \le m$ and every agent $i$, it holds that $v_i(e_j) \ge v_i(e_k)$.
\end{definition} 

The following lemma is an immediate corollary of a result of~\cite{bouveret2016characterizing}.

\begin{lemma}
\label{lem:BL}
Consider the class of additive valuations, any share for this class (e.g., MMS, $NS_{n,q}$) and any value of $\rho$ (including $\rho = 1$). Then any (polynomial time) allocation algorithm that gives every agent at least a $\rho$ fraction of her share {for} ordered instances, can be transformed into a (polynomial time) algorithm that produces such an allocation for arbitrary instances.
\end{lemma}

The proof of Lemma~\ref{lem:BL} is based on defining a picking order for the arbitrary instance, based on the allocation for an ordered version of the instance. See~\cite{bouveret2016characterizing} for further details.
}

\begin{lemma}
\label{lem:fixqlargen}
For the class of additive valuations and every $q \ge 1$, if $NS_{q,q}$  is feasible, then so is $NS_{n,q}$ for every $n > q$. Moreover, if an acceptable allocation for $NS_{q,q}$ can be found in polynomial time, so can an acceptable allocation for $NS_{n,q}$. 
\end{lemma}

\begin{proof}
Consider an instance with $n > q$ agents, and the associated $NS_{n,q}$ partitions for the $n$ agents.
{We assume that the instance is ordered. This assumption can be made without loss of generality, by Lemma~\ref{lem:BL}.}

Create $n$ bins. Place items $e_1, \ldots, e_{n-q}$ in bins $B_1, \ldots, B_{n-q}$, respectively. 
Now insert the minimum number of items in reverse order (starting at $e_m$, the minimum may be~0) into $B_1$ until there is an agent that {\em desires} $B_1$, {where we say that an agent \emph{desires a bin} if the agent's total value for the items in the bin is at least as large as her $NS_{n,q}$ share value}. Give $B_1$ to the respective agent (breaking ties arbitrarily if there is more than one such agent), and remove the agent and the bin. Observe that for each of the remaining agents, the $NS_{n-1,q}$ share with respect to the remaining items is at least as large as her original $NS_{n,q}$ share with respect to all items. (The bins $B_2, \ldots, B_n$ in the partition in $\FF_{n,q}$ that gave rise to $NS_{n,q}$ can serve as a partition in $\FF_{n-1,q}$ with respect to the remaining items. 
If some items of the original $B_1$ also remain, they can be moved into $B_2$, while still giving a partition in $\FF_{n-1,q}$.)  Continue in the same fashion with bin $B_2$, and more generally with the remaining bins, up to and including bin $B_{n-q}$. At this point only $q$ agents remain, and for each of them the $NS_{q,q}$ share over the remaining items has value at least as high as the original $NS_{n,q}$ share. {After computing the respective $NS_{q,q}$ partitions over these remaining items, the assumption that $NS_{q,q}$ {is feasible} implies that one can allocate to each of the $q$ remaining agents a bundle that she values at least as much as her share.} 

{The above algorithm runs in polynomial time (if an acceptable allocation for $NS_{q,q}$ can be found in polynomial time). 
This follows from \cref{lem:NS-poly-time} that shows that the $NS_{n,q}$ value of each of the agents can be computed in polynomial time, together with a respective $NS_{n,q}$ partition.}
\end{proof}

\begin{definition}
Consider two partitions of $\items$ into three disjoint bundles, $P^i = (B_1^i, B_2^i, B_3^i)$ and $P^j = (B_1^j, B_2^j, B_3^j)$. We say that $P^i$ and $P^j$ are {\em fully intersecting} if every bundle of $P^i$ intersects every bundle of $P^j$.
\end{definition}

\begin{lemma}\label{lem:no-inter-feasible}
Consider $n=3$, and for every agent $i \in \{1,2,3\}$ a partition of $\items$ into three disjoint bundles $P^i = (B_1^i, B_2^i, B_3^i)$. If no two partitions $P^i$ and $P^j$ are fully intersecting, then there is an allocation $A = (A_1, A_2, A_3)$ such that for every agent $i$ with an additive valuation function {$v_i$ it holds that} $v_i(A_i) \ge \min[v_i(B_1^i), v_i(B_2^i), v_i(B_3^i)]$. Moreover, given the partitions $P^i$ for $i \in \{1,2,3\}$ and the additive valuations {$(v_1,v_2,v_3)$,} 
such an allocation $A$ can be computed in polynomial time.
\end{lemma}

\begin{proof}
Consider three agents with additive valuation functions, $v_i$ for each agent $i$, and associated partitions $P^i = (B_1^i, B_2^i, B_3^i)$, for every agent $i \in \{1,2,3\}$.

For every two agents $i$ and $k$ (it can be that $i = k$) we define $v_i(P^k) = \min_{j \in \{1,2,3\}}[v_i(B_j^k)]$. We refer to $v_i(P^i)$ as the {\em share} of {agent $i$.}
We show that there is an allocation that gives every agent at least her share.

We say that a bundle $B$ is {\em acceptable}
for agent $i$ if 
$v_{i}(B) \ge v_i(P^i)$. 
Every bundle in her $P^i$ partition is acceptable for $i$. Moreover, for every $k \not= i$, at least one bundle in the $P^k$ partition is acceptable for $i$ (as $\max_{j \in \{1,2,3\}}[v_i(B_j^k)] \ge \frac{1}{3}\sum_{j \in \{1,2,3\}} v_i(B_j^k) = \frac{1}{3}\sum_{j \in \{1,2,3\}} v_i(B_j^i) \ge v_i(P^i)$). 

Consider two of the agents, say~1 and~2. Suppose that there are two bundles in the $P^2$ partition that are acceptable for agent~1. Then we can give agent~3 the bundle that is acceptable for her in the $P^2$ partition, give agent~1 a different bundle from the $P^2$ partition that is acceptable for agent~1, and give the remaining bundle to agent~2. This is an acceptable allocation. Hence we may assume that at most one bundle in $P^2$ is acceptable for agent~1.

{As at most one bundle in $P^2$ is acceptable for agent~1, there is a bundle $B^2_j$ in the $P^2$ partition that is not acceptable for agent~1. Suppose that there is a bundle $B^1_{j'}$ in the $P^1$ partition that is disjoint from $B^2_j$. 
}
In this case, if $B^2_j$ is removed, agent~1 can still make two bundles each of value at least $v_i(P^i)$, where one bundle is $B^1_{j'}$ and the other is $\items \setminus (B^1_{j'} \cup B^2_{j})$. Hence if $v_3(B^2_j) \ge v_3(P^3)$, we can give $B^2_j$ to agent~3, and use cut and choose (any of the agents~1 or~2 can serve as the cutter) to give each of the other two agents a bundle of value at least her share. Alternatively, if $v_3(B^2_j) < NS_{3,3}(v_3)$, we can give $B^2_j$ to agent~2, and let agent~3 choose among the two bundles, $B^1_{j'}$ or $\items \setminus (B^1_{j'} \cup B^2_{j})$, the one with higher $v_3$ value. This value is at least $PS_3(v_3)$, and hence at least her share {$v_3(P^3)$}. As argued above, the remaining bundle gives agent~1 at least her share $v_i(P^i)$. In any case, we get an acceptable allocation. 

Hence we can assume that each of the two bundles in the $P^2$ partition that are not acceptable for agent~1 intersect all three bundles in the $P^1$ partition.

As partitions $P^1$ and $P^2$ are assumed not to be fully intersecting, there must be two bundles (one from each partition) that do not intersect. {Moreover, we may assume that there is exactly one pair of bundles that do not intersect, and that  each of these bundles is the only bundle from that partition that is acceptable for the other agent. (In every other case, the above arguments show that an acceptable allocation exists.)} In other words, we {may assume} that {the pair of} acceptable bundles do not intersect.
A similar assumption holds for every pair of two different agents.

In a given input instance, if two different bundles in the $P^i$ partition of agent $i$ are acceptable, {each for a different other agent}, then we have an acceptable allocation. (For example, if $B^1_1$ is acceptable for agent~2 and $B^1_2$ is acceptable for agent~3, then $(B^1_3, B^1_1, B^1_2)$ is an acceptable allocation.) Hence each agent $i$ has in her $P^i$ partition a unique bundle that is acceptable for the other two agents. As acceptable bundles do not intersect, these three bundles are disjoint. Hence they form an acceptable allocation (each agent gets the acceptable bundle from her own partition). 
\end{proof}

\begin{lemma}\label{lem:not-fully}
{Consider an ordered allocation instance with three agents, and} let $P^i$ and $P^j$ be two $NS_{3,3}$ partitions of $\items$. Then $P^i$ and $P^j$ are not fully intersecting.
\end{lemma}

\begin{proof}
Let $e_1, \ldots, e_m$ be {an} 
ordering of the items (from most to least valuable) in the ordered instance. Given this order, an $NS_{3,3}$ partition is fully specified by four integral {\em cutting points} $0 \le c_1 \le c_2 \le c_3 \le c_4 \le m$, where $B_1 = \{e_1, \ldots, e_{c_1}\} \cup \{e_{c_4}, \ldots e_m\}$, $B_2 = \{e_{c_1}, \ldots, e_{c_2}\} \cup \{e_{c_3}, \ldots e_{c_4}\}$, and $B_3 = \{e_{c_2}, \ldots, e_{c_3}\}$. That is, $B_3$ is composed of a single interval, whereas each of $B_1$ and $B_2$ is a union of two intervals. (If there are equalities in the chain  $0 \le c_1 \le c_2 \le c_3 \le c_4 \le m$ of weak inequalities, then some of the intervals are empty.)

Each of the partitions $P^i$ and $P^j$ has four cutting points, and together these cutting points create $2*4 +1 = 9$ intervals (some of which might be empty). Each of these intervals {(when non-empty)} represents a consecutive set of items that in each of the two partitions belong to a single bundle. Hence there is a mapping between these {non-empty} intervals and pairs of bundles of $P^i$ and $P^j$ that intersect: every {non-empty} 
interval can be mapped to such an intersecting pair, and every intersecting pair is mapped to by some {non-trivial} interval. The mapping is surjective (onto) but cannot be one-to-one. In particular, both the first {interval} and the last interval {(when non-empty)} map to the same pair $(B_1^i, B_1^j)$. Hence the number of intersecting pairs is at most $9 - 1 = 8$, proving that $P^i$ and $P^j$ are not fully intersecting. 
\end{proof}

\begin{theorem}
\label{thm:NS3feasible}
For the class of additive valuations and every $n \ge 3$, the $NS_{n,3}$ share is feasible.
Moreover, an acceptable allocation can be found in polynomial time.
\end{theorem}

\begin{proof}
We prove the theorem for $n=3$, namely, for $NS_{3,3}$. By \cref{lem:fixqlargen} this will imply the theorem for all $n > 3$. {Moreover, we assume that the instance is ordered. This assumption can be made without loss of generality, by Lemma~\ref{lem:BL}.}

Let $P^i = (B_1^i, B_2^i, B_3^i)$ for $i\in \{1,2,3\}$ be an $NS_{3,3}$ partition of $\items$ for $v_i$. 
By \cref{lem:not-fully}, $P^i$ and $P^j$ are not fully intersecting for any $i,j\in \{1,2,3\}$.
By \cref{lem:no-inter-feasible} there is an allocation $A = (A_1, A_2, A_3)$ such that for every agent $i$ with an additive valuation function {$v_i$ it holds that} $v_i(A_i) \ge \min[v_i(B_1^i), v_i(B_2^i), v_i(B_3^i)] = NS_{3,3}(v_i)$. Thus, the share $NS_{3,3}$ is feasible. Moreover, {\cref{lem:no-inter-feasible} also shows that} given the partitions $P^i$ for $i \in \{1,2,3\}$ and the additive valuations {$(v_1,v_2,v_3)$,} such an allocation $A$ can be computed in polynomial time.
\end{proof}

\section{Discussion}\label{sec:discussion}

In this paper we have focused on fair allocation of indivisible \emph{goods} when agents have \emph{equal} entitlements to the goods, and have 
suggested notions for
share-based fairness.  

While we focus on \emph{equal} entitlements, the notions of a share, a share guarantee, and of a self-maximizing share, can naturally be extended to \emph{arbitrary} (non-equal) entitlements. 
In the definition of a share and of a share guarantee, one should simply replace the number of agents parameter $n$, by the entitlement of the agent.  
In \cref{app:ArbitraryEntitlements} we discuss some shares that were defined in prior work for arbitrary entitlements, and whether they are self maximizing or not. Note that our non-existence results for SM-dominant feasible shares clearly extend to this more general set of problems. 
{An interesting direction for future work is of extending our work to settings where agents have arbitrary entitlements.} 

One can also consider items that are  chores (bads) instead of goods. Our definitions naturally extend to chores as well. Our positive results (Section \ref{sec:additive})
are for agents with additive valuations over goods. The positive results are not based on picking orders, since it is impossible to obtain such results using picking order shares when items are goods. That is, for every $\rho > 0$, for sufficiently large $n$, no picking order share has value larger than $\rho$-MMS when valuations are additive over goods. This follows from results in~\cite{ABM2016}. However, for chores, there are picking orders that guarantee that every agent with an additive valuation gets a bundle of value no more than $\frac{5}{3}$ times her MMS. See~\cite{ALW20}. As our results that picking order shares are feasible, self maximizing and poly-time computable for additive valuations (see Section \ref{sec:picking}) do not rely on the items being goods, we get that for additive valuations over chores, there is a feasible share that is  $\frac{5}{3}$-dominating, self maximizing and poly-time computable.

Finally, we note that we leave open the problem of finding (for additive valuations, and also beyond this class), the largest value of $\rho$ for which there exists a feasible share that is  $\rho$-dominating, self maximizing and poly-time computable.

\bibliographystyle{alpha}

\bibliography{bib}

\newcommand{\etalchar}[1]{$^{#1}$}
\begin{thebibliography}{CKM{\etalchar{+}}19}

\bibitem[ABF{\etalchar{+}}21]{AmanatidisBFLLR21}
Georgios Amanatidis, Georgios Birmpas, Federico Fusco, Philip Lazos, Stefano
  Leonardi, and Rebecca Reiffenh{\"{a}}user.
\newblock Allocating indivisible goods to strategic agents: Pure nash
  equilibria and fairness.
\newblock In {\em Web and Internet Economics - 17th International Conference,
  {WINE} 2021, Potsdam, Germany, December 14-17, 2021, Proceedings}, volume
  13112 of {\em Lecture Notes in Computer Science}, pages 149--166. Springer,
  2021.

\bibitem[ABM16]{ABM2016}
Georgios Amanatidis, Georgios Birmpas, and Evangelos Markakis.
\newblock On truthful mechanisms for maximin share allocations.
\newblock In {\em Proceedings of the Twenty-Fifth International Joint
  Conference on Artificial Intelligence}, IJCAI'16, pages 31--37, 2016.

\bibitem[ALW20]{ALW20}
Haris Aziz, Bo~Li, and Xiaowei Wu.
\newblock Approximate and strategyproof maximin share allocation of chores with
  ordinal preferences.
\newblock {\em CoRR}, abs/2012.13884, 2020.

\bibitem[AMNS17]{amanatidis2017approximation}
Georgios Amanatidis, Evangelos Markakis, Afshin Nikzad, and Amin Saberi.
\newblock Approximation algorithms for computing maximin share allocations.
\newblock {\em ACM Transactions on Algorithms (TALG)}, 13(4):1--28, 2017.

\bibitem[ARSW17]{AzizRSW17}
Haris Aziz, Gerhard Rauchecker, Guido Schryen, and Toby Walsh.
\newblock Algorithms for max-min share fair allocation of indivisible chores.
\newblock In Satinder~P. Singh and Shaul Markovitch, editors, {\em Proceedings
  of the Thirty-First {AAAI} Conference on Artificial Intelligence, February
  4-9, 2017, San Francisco, California, {USA}}, pages 335--341. {AAAI} Press,
  2017.

\bibitem[BEF21a]{BEF2021c}
Moshe Babaioff, Tomer Ezra, and Uriel Feige.
\newblock Best-of-both-worlds fair-share allocations.
\newblock {\em arXiv preprint arXiv:2102.04909}, 2021.

\bibitem[BEF21b]{BEF2021b}
Moshe Babaioff, Tomer Ezra, and Uriel Feige.
\newblock Fair-share allocations for agents with arbitrary entitlements.
\newblock In {\em {EC} '21: The 22nd {ACM} Conference on Economics and
  Computation, Budapest, Hungary, July 18-23, 2021}, page 127. {ACM}, 2021.

\bibitem[BK20]{BK20}
Siddharth Barman and Sanath~Kumar Krishnamurthy.
\newblock Approximation algorithms for maximin fair division.
\newblock {\em ACM Transactions on Economics and Computation (TEAC)},
  8(1):1--28, 2020.

\bibitem[BL16]{bouveret2016characterizing}
Sylvain Bouveret and Michel Lema{\^\i}tre.
\newblock Characterizing conflicts in fair division of indivisible goods using
  a scale of criteria.
\newblock {\em Autonomous Agents and Multi-Agent Systems}, 30(2):259--290,
  2016.

\bibitem[BNT20]{BabaioffNT2020}
Moshe Babaioff, Noam Nisan, and Inbal Talgam{-}Cohen.
\newblock Competitive equilibrium with indivisible goods and generic budgets.
\newblock {\em Mathematics of Operations Research}, 2020.

\bibitem[Bud11]{Budish11}
Eric Budish.
\newblock The combinatorial assignment problem: Approximate competitive
  equilibrium from equal incomes.
\newblock {\em Journal of Political Economy}, 119(6):1061--1103, 2011.

\bibitem[CFS17]{CFS17}
Vincent Conitzer, Rupert Freeman, and Nisarg Shah.
\newblock Fair public decision making.
\newblock In Constantinos Daskalakis, Moshe Babaioff, and Herv{\'{e}} Moulin,
  editors, {\em Proceedings of the 2017 {ACM} Conference on Economics and
  Computation, {EC} '17, Cambridge, MA, USA, June 26-30, 2017}, pages 629--646.
  {ACM}, 2017.

\bibitem[CKM{\etalchar{+}}19]{CKMPSW19}
Ioannis Caragiannis, David Kurokawa, Herv{\'e} Moulin, Ariel~D Procaccia,
  Nisarg Shah, and Junxing Wang.
\newblock The unreasonable fairness of maximum nash welfare.
\newblock {\em ACM Transactions on Economics and Computation (TEAC)},
  7(3):1--32, 2019.

\bibitem[FGH{\etalchar{+}}19]{farhadi2019fair}
Alireza Farhadi, Mohammad Ghodsi, Mohammad~Taghi Hajiaghayi, Sebastien Lahaie,
  David Pennock, Masoud Seddighin, Saeed Seddighin, and Hadi Yami.
\newblock Fair allocation of indivisible goods to asymmetric agents.
\newblock {\em Journal of Artificial Intelligence Research}, 64:1--20, 2019.

\bibitem[FN22]{FN22}
Uriel Feige and Alexey Norkin.
\newblock Improved maximin fair allocation of indivisible items to three
  agents.
\newblock {\em CoRR}, abs/2205.05363, 2022.

\bibitem[FST21]{FST21}
Uriel Feige, Ariel Sapir, and Laliv Tauber.
\newblock A tight negative example for {MMS} fair allocations.
\newblock {\em CoRR}, abs/2104.04977, 2021.

\bibitem[GHS{\etalchar{+}}18]{GhodsiHSSY18}
Mohammad Ghodsi, Mohammad~Taghi Hajiaghayi, Masoud Seddighin, Saeed Seddighin,
  and Hadi Yami.
\newblock Fair allocation of indivisible goods: Improvements and
  generalizations.
\newblock In {\em Proceedings of the 2018 {ACM} Conference on Economics and
  Computation, Ithaca, NY, USA, June 18-22, 2018}, pages 539--556. {ACM}, 2018.

\bibitem[GLW21]{GLW21}
Laurent Gourv{\`{e}}s, Julien Lesca, and Ana{\"{e}}lle Wilczynski.
\newblock On fairness via picking sequences in allocation of indivisible goods.
\newblock In Dimitris Fotakis and David~R{\'{\i}}os Insua, editors, {\em
  Algorithmic Decision Theory - 7th International Conference, {ADT} 2021,
  Toulouse, France, November 3-5, 2021, Proceedings}, volume 13023 of {\em
  Lecture Notes in Computer Science}, pages 258--272. Springer, 2021.

\bibitem[GM19]{GM19}
Laurent Gourv{\`{e}}s and J{\'{e}}r{\^{o}}me Monnot.
\newblock On maximin share allocations in matroids.
\newblock {\em Theor. Comput. Sci.}, 754:50--64, 2019.

\bibitem[GMT19a]{garg2019approximating}
Jugal Garg, Peter McGlaughlin, and Setareh Taki.
\newblock Approximating maximin share allocations.
\newblock {\em Open access series in informatics}, 69, 2019.

\bibitem[GMT19b]{GMT19}
Jugal Garg, Peter McGlaughlin, and Setareh Taki.
\newblock Approximating maximin share allocations.
\newblock In Jeremy~T. Fineman and Michael Mitzenmacher, editors, {\em 2nd
  Symposium on Simplicity in Algorithms, {SOSA} 2019, January 8-9, 2019, San
  Diego, CA, {USA}}, volume~69 of {\em {OASICS}}, pages 20:1--20:11. Schloss
  Dagstuhl - Leibniz-Zentrum f{\"{u}}r Informatik, 2019.

\bibitem[GT20]{GT20}
Jugal Garg and Setareh Taki.
\newblock An improved approximation algorithm for maximin shares.
\newblock In {\em Proceedings of the 21st ACM Conference on Economics and
  Computation}, pages 379--380, 2020.

\bibitem[HL19]{huang2019algorithmic}
Xin Huang and Pinyan Lu.
\newblock An algorithmic framework for approximating maximin share allocation
  of chores.
\newblock {\em arXiv preprint arXiv:1907.04505}, 2019.

\bibitem[HS21]{HosseiniSearns21}
Hadi Hosseini and Andrew Searns.
\newblock Guaranteeing maximin shares: Some agents left behind.
\newblock In Zhi{-}Hua Zhou, editor, {\em Proceedings of the Thirtieth
  International Joint Conference on Artificial Intelligence, {IJCAI} 2021,
  Virtual Event / Montreal, Canada, 19-27 August 2021}, pages 238--244.
  ijcai.org, 2021.

\bibitem[KPW18]{KurokawaPW18}
David Kurokawa, Ariel~D. Procaccia, and Junxing Wang.
\newblock Fair enough: Guaranteeing approximate maximin shares.
\newblock {\em J. {ACM}}, 65(2):8:1--8:27, 2018.

\bibitem[LMMS04]{LiptonMMS04}
Richard~J. Lipton, Evangelos Markakis, Elchanan Mossel, and Amin Saberi.
\newblock On approximately fair allocations of indivisible goods.
\newblock In {\em Proceedings of the 5th {ACM} Conference on Electronic
  Commerce}, EC'04, pages 125--131, 2004.

\bibitem[Pro20]{Procaccia20}
Ariel~D. Procaccia.
\newblock An answer to fair division's most enigmatic question: technical
  perspective.
\newblock {\em Commun. {ACM}}, 63(4):118, 2020.

\end{thebibliography}

\appendix

\section{Notes on Auxiliary Goals}\label{app:auxiliary}

It is important to clarify the following point. Though it is interesting to know what the highest possible value of $\rho$ is in Question~\ref{Q2}, this does not necessarily imply that one would want to actually use a share that has the highest possible value of $\rho$.  It is true that shares with high value of $\rho$ improve the guarantees given to agents. However, it is also true that the improved guarantees limit the flexibility given to the allocator -- he has fewer {acceptable allocations}  to choose from. This might hurt other measures of performance that the allocator cares about, and possibly also the agents care about. We refer to these other measures of performance as {\em auxiliary goals}, as they go beyond the goal of giving every agent at least her share. Here are some examples of auxiliary goals. 

\begin{itemize}

    \item Auxiliary goals such as maximizing welfare (the sum of values received by all agents) or maximizing Nash Social Welfare (the product of values received by all agents) attempt to choose a ``best possible'' allocation (according to the auxiliary goal), subject to the constraint that every agent gets at least her share. 
  
  \item {\em Ex-ante guarantees}. Allocation processes are often randomized. For example, when allocating a single indivisible item to one of two agents, having a lottery to decide who receives the item is common practice. 
  The lottery offers the agents an ex-ante guarantee that is valuable (for example, that the agent has probability $\frac{1}{2}$ of winning the item), whereas no ex-post guarantee can be given to both agents. 
  Hence an allocator may have the auxiliary goal of selecting a distribution over allocations that offers some desirable ex-ante guarantee. The ``contracts'' with the agents then require that this distribution be supported only over allocations that give each agent at least her share. 
  
  \item Mechanisms with money. Though the interest in this paper is in settings without money, the notion of a share may be relevant also to mechanisms with money. Suppose for example that the items are health services (such as doctor appointments, or medical drugs) and the agents are members of a health maintenance organization (HMO) that provides these services. Every member, as a function of her needs (her valuation function), is entitled to at least some minimum service (her share). Suppose now that members can ask for additional services beyond their share, as long as they are willing to pay for these additional services, and as long as the shares of other members are respected. In this case, the auxiliary goal of the HMO when running the allocation mechanism might be to maximize revenue (choosing the allocation of resources that generates maximum revenue from the members), subject to the constraint that every member gets at least her share, and can be charged only for services that are beyond her share. 

\end{itemize}

As the examples above show, there may be some tension between the desire to have shares of high value and the desire to accomplish additional auxiliary goals. 
High share values protect each individual agent against the possibility that the auxiliary goals will drive her value to be too low. High share values also reduce the motivation of agents to try to manipulate the allocation mechanism, as there is less to gain by manipulation, and potentially more to lose. However, {share values that are too high} 
might limit the extent to which the auxiliary goals can be achieved, and hence might also not be desirable. If auxiliary goals are present, then they should be taken into account when deciding on a choice of share function, so as to strike a good balance between the share guarantees and the auxiliary goals. In the current paper we are only concerned with determining what values of $\rho$ are achievable by $\rho$-dominating shares, and leave aspects concerning tradeoffs with other auxiliary goals to future work. 

{We end this section by noting that the standard economic efficiency auxiliary goal of producing a Pareto optimal allocation is never in conflict with the selection of a feasible share.} For every feasible share, there always  an acceptable allocation that is also Pareto optimal. This is because replacing an acceptable allocation by an allocation that Pareto dominates it preserves acceptability.

\section{Statement of Intent}\label{sec:declaration}

Our focus in this paper is on the \emph{guarantee} that a feasible share provides to each agent. This should be distinguished from what we shall refer to as a \emph{non-binding ``statement of intent"}. Such a statement might be, for example, an intent to implement an allocation that gives every agent her maximin share. Often this statement can be met, but sometimes not (either because no maximin allocation exists, or because of computational difficulties of finding such an allocation even when it does exist, or because of computational difficulties of computing the maximin share values of the agents and verifying whether a given allocation satisfies them).
The statement of intent is non-binding and is perceived as such by the agents. 

Our paper concerns with a \emph{guarantee} (not just a non-binding statement of intent) that the implemented allocation will give the agent at least her share, for some share that is different than the maximin share and is feasible. Let us call our share the FSMP-share, where FSMP stands for feasible. self maximizing, and polynomial time computable. The actual value of the FSMP-share can be computed in polynomial time (so the agent can verify that the guarantee is met), and the share is self maximizing (reporting the true valuation function maximizes the value of the guarantee). Moreover, the FSMP-share is at least a $\rho$-fraction of the MMS. (Our current value of $\rho$ for every $n$ is $2/3$, but this ratio might be improved in the future.)

Using the above terminology, one might wonder what would an agent prefer, a statement of intent that is ``simple-to-understand" (e.g., concerning the maximin share), or a guaranteed FSMP-share? Our answer is that the agent would prefer not to be forced to choose between the two, but rather to get both. This is exactly what we offer the agent. The FSMP-share does not dictate which allocation algorithm is to be used. (On the contrary - the property of a share being self maximizing is more closely connected to incentive compatibility if the agent does not know which allocation algorithm will be run. See discussion on incentive compatibility in \cref{sec:contracts}.) The polynomial time allocation algorithm that we offer that guarantees every agent her FSMP-share only serves as a ``safety" allocation algorithm. However, in practice, the social planner can (and is advised to) additionally run any number of allocation algorithms of his choice (possibly, including allocation algorithms that search for a maximin allocation). For every allocation produced in any of the runs, the social planner can check in polynomial time whether every agent got at least her FSMP-share. Among the allocations that survive this test (the safety algorithm guarantees that at least one allocation survives), the social planner is allowed to choose any surviving allocation. If the social planner made a statement of intent with respect to MMS allocations and determines that some surviving allocations is an MMS allocation, he may choose such an allocation (fulfilling the statement of intent).
In any case, whichever allocation is chosen, it will give every agent at least her FSMP-share (satisfying the guarantee). 

We believe that in practice agents will appreciate getting an FSMP-share guarantee. (The actual FSMP-share of choice may be from the $NS_{n,q}$ family of shares introduced in this paper, or some other FSMP-share that is yet to be introduced.) This guarantee is easy to understand, in a sense that there is a polynomial time algorithm that translates it to a guaranteed value. In addition, the agent may also appreciate getting a ``statement of intent" of some sort (e.g., that an effort will be made to find a maximin fair allocation and such an allocation will often be implemented). Both FSMP-share guarantees and statements of intent have value in practice, and they can be taken into account simultaneously.

\section{Self-Maximizing Shares for Arbitrary Entitlements}
\label{app:ArbitraryEntitlements}

Two share notions (among others) defined for arbitrary entitlements are the  
Truncated-Proportional Share (TPS)~\cite{BEF2021c} and 
the AnyPrice Share (APS)~\cite{BEF2021b}. We next observe that while the APS is self maximizng, the TPS is not. 

\begin{proposition}\label{prop:shares-self-unequal}
    The APS is self-maximizing. 
    In contrast, the TPS is not self-maximizing, not even for additive valuation functions.
\end{proposition}

\begin{proof}
The APS is self-maximizing: given any price vector, the affordable bundle of highest value according to $v'$ is among the bundles considered when the APS chooses the bundle with highest value according to $v$, and hence no better than that bundle.

The example presented in the proof of Proposition \ref{prop:shares-self} which shows that the proportional share is not self-maximizing also proves that the TPS is not self-maximizing, as the TPS and PS are the same for that example. 
\end{proof}

\section{No More than the Share Guarantee}\label{sec:only-guaranteed}

{In this section we show that for most valuation classes of interest (e.g., unit demand, additive, submodular, subadditive), it is indeed the case that an agent might only get her guaranteed share (and no more) when the allocation and valuations of others are worse case. }

For the purpose of Proposition~\ref{pro:SM} below, we define the notion of a hereditary class of valuation functions. 

\begin{definition}
We say that a class $\cal{C}$ of valuation functions is {\em hereditary} if for every two sets of items, $\items' \subset \items$, and for every valuation function $v' \in \cal{C}$ for $\items'$, the valuation function $v$ for $\items$, defined as $v(S) = v'(S \cap \items')$ {for every $S\subseteq\items$}, is also in $\cal{C}$.
\end{definition}

Most valuation classes of interest (e.g., unit demand, additive, submodular, subadditive) are hereditary. An example of a valuation function that is not hereditary is an additive valuation function in which the value of every item is required to be a strictly positive integer.

\begin{proposition}
\label{pro:SM}
    Let $\cal{C}$ be a hereditary class of valuation functions, and let $s$ be a feasible share for $\cal{C}$, with associated share guarantee $\hat{s}$. Then for every set $\items$ of items and every number $n \ge 2$ 
    of agents, for every valuation $v \in \cal{C}$ defined over $\items$, there  is a feasible allocation mechanism $AM_v$ for $\cal{C}$ and there are valuations in $\cal{C}$ for the remaining $n-1$ agents, such that $AM_v$ allocates to the agent holding valuation function $v$ a bundle of items of value exactly $\ShareGfullv$.
\end{proposition}

\begin{proof}
Let $B\subseteq \items$ be such that $v(B) = \ShareGfullv$. 
{Without loss of generality we need to prove the claim for agent~1.} 
For agent~1, let $v_1 = v$, and for every other agent $i > 1$, let $v'_i\in \cal{C}$ be an arbitrary valuation function defined over $\items' = \items \setminus B$, and let $v_i$ be such that $v_i(S) = v'_i(S \cap \items')$ {(note that as $\cal{C}$ is a hereditary class of valuation functions there exist such $v_i$ in $\cal{C}$)}. 
Let $AM$ be an arbitrary feasible allocation mechanism for $\cal{C}$, and let $(A_1,\ldots, A_n)$ be its allocation on this given instance. Modify $AM$ to $AM_v$, as follows. Only in this given instance, replace the allocation to $(B, A_2 \setminus B, \ldots, A_n \setminus B)$.
{As for any $i>1$ it holds that $v_i(A_i \setminus B) = v_i(A_i)$,}
the allocation still gives every agent at least her share value, and hence is feasible. Thus $AM_v$ is a feasible allocation mechanism, and on this instance, it gives the agent holding $v$ a bundle of value exactly $\ShareGfullv$.
\end{proof}

\begin{remark}
The condition that the class of valuation functions is hereditary is not strictly necessary for Proposition~\ref{pro:SM}. Weaker conditions also suffice, though may require a different proof. On the other hand, if conditions stronger than just being hereditary are imposed, the proposition can be strengthened by reversing the order of quantifiers -- choosing one universal allocation mechanism $AM$ that applies for all $v$. For example, if $\cal{C}$ is the class of additive valuation functions, than $AM$ can be chosen as an allocation mechanism that maximizes welfare, conditioned on giving every agent at least her fair share. Then, for every additive valuation function $v$, one can choose additive valuation functions for the remaining agents, such that the allocation mechanism will necessarily give the agent holding $v$ a bundle of value exactly equal to the corresponding share guarantee. Further details are omitted. 
\end{remark}

\section{Missing Proof from Section \ref{sec:maximal}}\label{app:maximal}

\subsection{Missing Proof from Section \ref{sec:notion-self-max}}\label{app:notion-self-max}

\subsubsection{Proof of \cref{prop:shares-self}}\label{app:prop-shares-self}

We restate and prove  \cref{prop:shares-self}.

\propSelfMaxShares*

\begin{proof}
Fix any class $\CC$ of valuations. 
The MMS is self-maximizing {for $\CC$: for any valuation $v\in \CC$,} any MMS partition that is computed with respect to $v'\in\CC$ is also a legitimate candidate for an MMS partition for $v$, and hence no better than the actual best MMS partition with respect to $v$. 

{The proportional share is self maximizing for additive valuations if $n=2$: for additive valuation function $v$, the proportional share {guarantee} is $\min_{\{B \; | \; v(B) \ge \frac{1}{2}v(\items)\}}[v(B)]$. Let $B^*$ be a bundle achieving the minimum in this expression. Then for every additive valuation function $v'$, either $B^*$ or $\items \setminus B^*$ gives the proportional share {with respect to $v'$} (because $n=2$), {so at least one of these sets is considered acceptable when reporting $v'$}. 
As $v(\items \setminus B^*) \le v(B^*)$, the {implied} guarantee cannot improve by reporting $v'$ instead of $v$.}
{Thus, for any $v'$: $\ShareGvv = \ShareGfullv = v(B^*)=\max \{v(B^*), v(\items \setminus B^*)\}\geq 
\ShareG$ }

The proportional share (PS) is not self-maximizing for {additive valuations and} $n \ge 3$: let $n \ge 3$ and $\items = \{e_1, \ldots, e_{n+2}\}$, with $v(e_j) = \frac{n-1}{n}$ for $j \le n$, and $v(e_{n+1}) = v(e_{n+2}) = \frac{1}{2}$. The proportional share (PS) and the PS guarantee (the bundle $\{e_{n+1}, e_{n+2}\}$) are both~1. However, reporting $v'(e_j) = \frac{2n-1}{2n}$ for $j \le n$, and $v'(e_{n+1}) = v'(e_{n+2}) = \frac{1}{4}$, 
still have that the proportional share {with respect to $v'$} is $1$. However, the bundle $\{e_{n+1}, e_{n+2}\}$ no longer gives the PS-guarantee with respect to 
$v'$ (as its value with respect to $v'$ is less than $1$), and any acceptable bundle with respect to $v'$ contains at least one item $e_j$ with $j \le n$, and has at least two items. 
Consequently, {by reporting $v'$}
the PS-guarantee with respect to $v$  
improves to $\frac{n-1}{n} + \frac{1}{2} > 1$, where the inequality holds for $n \ge 3$.
Formally, for $v$ and $v'$ it holds that  $\ShareGvv = \ShareGfullv = 1<\frac{n-1}{n} + \frac{1}{2}=
\ShareG$, {and thus the proportional share for $n\geq 3$ is not self maximizing}. 

$\rho$-MMS is not self-maximizing for {additive valuations and} any $\rho$ such that $0 < \rho < 1$. We consider different values of $\rho$ {and for each present an instance with a beneficial manipulation}:
\begin{itemize}
    \item $\frac{1}{2} < \rho < 1$: There are $n$ agents and $m = n+1$ items, $v_i(e_j) = 1$ for $1 \le j \le n-1$,  $v_i(e_n) = \rho$, and $v_i(e_{n+1}) = 1 - \rho$. The MMS is~1, and the $\rho$-MMS guarantee may give the agent only item $e_n$, of value $\rho < 1$. 
    {By reporting value of $\frac{1}{2}$ for each of the items $e_n$ and $e_{n+1}$, }
    the MMS remains~1, but neither $e_n$ alone nor $e_{n+1}$ alone have reported value at least~$\rho$. Every bundle that has reported value at least $\rho$ has true value at least~1, which is larger than $\rho$. Hence, reporting $v_i$ does not maximize the {implied} guarantee of agent $i$.
    \item $0 < \rho < \frac{1}{2}$: The above example extends to every $0 < \rho < \frac{1}{2}$, by reporting values of $\frac{\rho}{2}$ and $1 - \frac{\rho}{2}$ instead of $\rho$ and $1 - \rho$, for items $e_n$ and $e_{n+1}$, respectively. ({By misreporting,} the agent improves the {implied} guarantee from $\rho < \frac{1}{2}$ to $1 - \rho > \frac{1}{2}$.) 
    \item For the case $\rho = \frac{1}{2}$ and $n \ge 2$, consider an example with $m = n+2$ items, $v_i(e_j) = 1$ for $1 \le j \le n-2$,  $v_i(e_{n-1}) = v_i(e_{n}) = \frac{3}{4}$, and $v_i(e_{n+1}) = v_i(e_{n+2}) = \frac{1}{4}$. By reporting values $\frac{1}{8}$ instead of $\frac{1}{4}$ and $\frac{7}{8}$ instead of $\frac{3}{4}$, the {implied} guarantee improves from $\frac{1}{2}$ to $\frac{3}{4}$.
\end{itemize} 
\end{proof}

\subsubsection{Proof of \cref{prop:self-max-properties}}\label{app:prop-self-max-properties}
We restate and prove  \cref{prop:self-max-properties}.
\propSelfMaxProperties*
\begin{proof}
Suppose that $s$ is a share that is not scale invariant {for $\CC$}. Then there is a valuation function $v\in \CC$ and a multiplicative factor $c>0$ for which the share guarantee satisfies
$\ShareGcvcv=\ShareGfullcv\neq c\cdot \ShareGfullv = c\cdot \ShareGvv$. 
Assume that $\ShareGcvcv> c\cdot \ShareGvv$.
Then $s$ is not self maximizing, because for an agent with valuation function $v$, by reporting $c \cdot v$ (which is in $\CC$ by assumption) the agent guarantees to herself a bundle of higher value (according to her original $v$) than by reporting $v$. Formally, 
$c\cdot \ShareGvcv  = \ShareGcvcv> c\cdot \ShareGvv$.

Suppose that $s$ is a share that is not monotone {for $\CC$}. Then there are two valuation functions, $v,v'\in \CC$, such that $v \ge v'$ but 
$\ShareGvv= \ShareGfullv< \ShareGfullvp =\ShareGvpvp$. 
Then $s$ is not self maximizing, because for an agent with valuation function $v\in \CC$, by reporting $v'\in \CC$ she guarantees to herself a bundle of higher value (according to $v'$) than the value guaranteed (according to $v$) by reporting $v$. As $v > v'$, this strict inequality also holds when considering the value of bundles according to $v$ when the report is $v'$. 
{Formally: if $v,v'\in \CC$ such that $v \ge v'$ but $\ShareGvv < \ShareGvpvp$, then we get a contradiction to $s$ being self maximizing. An agent with value $v\in \CC$ can increase her share guarantee by reporting $v'\in \CC$ as $\ShareGvv < \ShareGvpvp\leq \ShareG$, where the second inequality follow from monotonicity ($v(S)\geq v'(S)$ for any set $S$).}

Suppose that $s$ is a share that is not $1$-Lipschitz {for $\CC$}. Then there is a value $\delta > 0$ and two valuation functions, $v,v'\in \CC$, such that for every bundle $S\in \items$ it holds that $|v(S) - v'(S)| \le \delta$, but $\ShareGvpvp = \ShareGfullvp> \ShareGfullv +\delta = \ShareGvv+\delta$.  
Then $s$ is not self maximizing, because for an agent with valuation function $v\in \CC$, by reporting $v'\in \CC$ she guarantees to herself a bundle of value  (according to $v'$) higher by more than $\delta$ than the value guaranteed (according to $v$) by reporting $v$. As for every bundle the values of $v'$ and $v$ differ by at most $\delta$, then when values are measured by respect to $v$, by reporting $v'$ the agent guarantees to herself a bundle of value  strictly higher than the value guaranteed by reporting $v$.
Formally, if $|v(S) - v'(S)| \le \delta$ but $\ShareGvpvp > \ShareGvv+\delta$ then $\ShareG + \delta\geq \ShareGvpvp > \ShareGvv+\delta$, and thus $\ShareG > \ShareGvv$, in contradiction to the share being self maximizing {for $\CC$}. 
\end{proof}

\subsection{Missing Proof from Section \ref{sec:dominant-feasible}}\label{app:dominant-feasible}

\subsubsection{Proof of \cref{prop:feasibleMMS-maximal}}\label{app:prop-feasibleMMS-maximal}
We restate and prove \cref{prop:feasibleMMS-maximal}.
\propMMSSMdominating*
\begin{proof}
The MMS dominates every feasible share $s$ for $\CC$, as if $s$ is strictly larger than the MMS for some valuation $v$ in $\CC$, it cannot be feasible because the MMS is defined to be 
the largest value that can be concurrently given to each of the $n$ agents when all have that valuation $v$. Thus, the MMS is a dominant share for $\CC$. If the MMS is also feasible for class $\CC$, then it is a dominant feasible share (and thus also an SM-dominant feasible share) for $\CC$. 

We next prove that {when the MMS is feasible, it}  
is the unique dominant feasible share, and also the unique SM-dominant feasible share. 
Consider any  dominant feasible share $s$ for class $\CC$. We have seen that the MMS dominates $s$ (as $s$ is feasible), and as $s$ dominates the MMS (as $s$ is dominant and the MMS is feasible by assumption), $s$ must equal the MMS, and thus the MMS is the unique dominant feasible share for $\CC$.   

As the MMS is self maximizing {for any $\CC$} (\cref{prop:shares-self}), if the MMS is feasible, then every SM-dominant feasible share for $\CC$ must dominate the MMS, and thus must equal the MMS. Thus, the MMS is the unique SM-dominant feasible share for $\CC$.  
\end{proof}

\section{Missing Proof from Section \ref{sec:no-dominant}}\label{app:no-dominant}

\subsection{Missing Proof from Section \ref{sec:picking}}\label{app:picking}
\subsubsection{Proof of \cref{prop:picking}}\label{app:prop-picking}

\cref{prop:picking} directly follows from the following three lemmas.

{When valuations are not additive, picking order shares need not be an ordinal maximin share, but nevertheless, we show that they are self maximizing. }

\begin{lemma}\label{prop:picking-well-behaved} 
Fix any class $\CC$ of valuations. 
For any picking order $\omega$, the $\omega$-Picking-Order share is self maximizing for $\CC$.
\end{lemma}

\begin{proof}
{Given a set $\items$ of items, a picking order $\omega$ and an identity $i$ for the agent {that has valuation $v_i$}, consider the zero-sum game with two players: the agent and the adversary. The agent gets to pick items in her turns according to $\omega$, and the adversary picks items in the remaining turns. The {outcome} of the game for the agent is the value (according to $v_i$) of the bundle that she picked.

As the game is finite, there are only finitely many strategies for the players. Every strategy $\sigma$ for the agent defines a collection of feasible bundles, namely, those bundles for which there is some strategy of the adversary that leads to the agent getting the bundle. The value of strategy $\sigma$ for the agent is the smallest value of any of the respective feasible bundles. The value  of {this zero-sum}  
game is the highest value among the values of the strategies of the agent, and is equal to the value of the share.

Suppose that the agent reports a valuation function $v'_i$ instead of the true $v_i$. Then the {implied guarantee} 
with respect to $v'_i$ is determined by some specific strategy $\sigma'$ in the above game. All {acceptable}  
bundles with respect to $\sigma'$ are bundles that an agent with valuation function $v'_i$ values at least as much as her share {(with respect to $v'_i$)}, and hence might be allocated to agent $i$ by an allocation algorithm that guarantees the share value. Hence the value that this guarantees to agent $i$ is not larger than the minimum value according to $v_i$ of a bundle that is {acceptable}  
for $\sigma'$. As $\sigma'$ is among the strategies considered when reporting $v_i$, and the share value for $v_i$ is computed according to the best strategy, it follows that the {implied guarantee} 
when reporting the true $v_i$ is at least as large as when reporting $v'_i$.}
\end{proof}

\begin{lemma}\label{prop:picking-feasible} 
{Fix any class $\CC$ of valuations. 
For any picking order $\omega$, the $\omega$-Picking-Order share is feasible for $\CC$}.  
Moreover, if valuations are additive then an acceptable allocation can be computed in polynomial time.  
\end{lemma}
\begin{proof}
{Every} agent simply follows the same {strategy} that defines her picking-order share. {Thus all agents simultaneously are guaranteed that each receives an acceptable bundle. }
When {an agent} has an additive valuation, that {strategy} simply picks a highest value available item at each picking point, 
{thus the allocation is poly-time computable.}   
\end{proof}

Crucial to the proof of our main theorem is the next claim, showing that for additive valuations,  for some picking order the share is as large as the MMS.

\begin{lemma}\label{prop:picking-MMS} 
	Consider a setting with $n$ agents and a set of items $\items$. 
	{For any additive valuation $v$ there exists a picking order $\omega_v$ for which the $\omega_v$-Picking-Order share for an agent with valuation $v$ equals to the MMS of that agent, that is, to $\MMSv$. }
\end{lemma}
\begin{proof} 
Consider any partition of the items to $n$ sets (parts) such that the value of the lowest-value part equals to {$\MMSv$.} 
For the given additive valuation, order the items in non-increasing value order, breaking ties arbitrarily. 
Name each item by its place in the order. Then define a picking order {$\omega_v$} 
using the MMS partition with each part associated with one of the agents, and she gets to pick at the rounds that are defined by the names of the items in that part.  Every time an agent needs to pick, she picks an available item that is of highest value, and due to additivity, she gets a value at least as high as the value of her part. Thus the agent that gets the lowest value gets at least her MMS (and no more, by the MMS definition).   
\end{proof}

\section{Missing Proofs from Section \ref{sec:domination-SM}}\label{app:dominatio}
\subsection{Proof of \cref{thm:ICdominates}}\label{app:thm:ICdominates}
\thmFeasibleToSMFeasible*

\begin{proof}
All valuation functions throughout the proof are assumed to be from class $\CC$, and domination, feasibility and self-maximization of shares is with respect to valuation function in class $\CC$.

As $s$ is feasible, for every $m$ and $n$ there is an allocation algorithm $A_{m,n}$ (though not necessarily a computationally efficient one) that for every instance with $m$ items and $n$ agents provides an {acceptable allocation: one} that gives every agent $i$ with valuation $v_i$ a bundle of value at least $\Sharefullvi$. Fix such an algorithm, $A_{m,n}$. Without loss of generality, we shall assume that given the names of the agents and the items, algorithm $A_{m,n}$ is deterministic. (For example, if $A_{m,n}$ was randomized, we fix all its coin tosses in advance.)  Algorithm $A_{m,n}$ partitions the valuation functions of $\CC$ into at most $2^{2^m}$ sub-classes as follows. Each sub-class is represented by a binary vector {$x\in \{0,1\}^{2^m}$} of length $2^m$, where each entry in this vector is indexed by a distinct bundle {$B$} of items in $\items$. Given such a binary vector $x$, valuation function $v$ belongs to the subclass associated with $x$ if and only if the following condition holds:

\begin{itemize}
    \item For every bundle $B \subset \items$, the corresponding entry $x_B$ in $x$ is~1 if and only if there is some setting of valuation functions for all agents in which at least one agent (say agent $i$) has valuation function $v$, and $A_{m,n}$ gives this agent $i$ the bundle $B$. 
\end{itemize}

The collection of all subclasses is denoted by $\cal{S}$. For valuation function $v$, we let $x^v$ denote the subclass that contains $v$.

For a valuation function $v$ and a subclass represented by $x$ (regardless if $v$ belongs to this subclass or not), we define $v(x) = \min_{\{B \; | \; x_B = 1\}}[v(B)]$. We define the share $s'$ as
$\Sharefullpv = \max_{x \in {\cal{S}}}[v(x)]$.
{Note that as $s'$ is defined to be a value of some bundle, its share guarantee $\hat{s}'$ is the same as the share itself.}

We need to show three properties for $s'$.

\begin{enumerate}

\item $s'$ is feasible. This is established by the following allocation algorithm $A'_{m,n}$ (which is not necessarily computationally efficient, not even if $A_{m,n}$ is computationally efficient). Given valuation functions $v_1, \ldots, v_n$, for every valuation function $v_i$, algorithm $A'_{m,n}$ determines the class $x$ that maximizes $v_i(x)$ (breaking ties arbitrarily), and replaces $v_i$ by an arbitrary valuation function $v'_i$ from this class $x$ ($x$ is the class that gives rise to the share value $\Sharefullpv$). Thereafter, $A'_{m,n}$ runs $A_{m,n}$ on inputs $v'_1, \ldots, v'_n$.

\item $s'$ is self maximizing. Suppose that one reports $v'$ instead of the true $v$. Then there is some subclass $z$  such that $\Sharefullpvp = v'(z) = \min_{\{B \; | \; z_B = 1\}}[v'(B)]$. {All bundles $B$ for which $z_B = 1$ are acceptable bundles for $v'$.} Hence {the implied guarantee for $v$ with respect to the report $v'$} is not more than $\min_{\{B \; | \; z_B = 1\}}[v(B)] = v(z)$. As $\hat{s}'_{v}(v) =\Sharefullpv {=\max_{y \in {\cal{S}}}[v(y)]}\ge v(z) \ge \hat{s}'_{v}(v')$ {(the last inequality is not necessarily an equality, as there might be bundles $B$ acceptable for $v'$ beyond those for which $z_B = 1$)},
reporting $v'$ instead of $v$ does not increase the implied guarantee.

    \item $\Sharefullpv \ge \Sharefullv$ for every valuation function $v$. This is established by the following chain of inequalities:

$$\Sharefullpv = \max_{x \in {\cal{S}}}[v(x)] \ge v(x^v) \ge \Sharefullv$$ 
where the last inequality holds because $s$ is feasible.
\end{enumerate}
\end{proof}

\section{Missing Proof from Section \ref{sec:additive}}\label{app:additive}
\subsection{Missing Proof from Section \ref{sec:additive-ordinal}}\label{app:additive-ordinal}

\begin{observation}
\label{obs:ordinal-feasible}
{For two agents with additive valuations, for any family $\FF_2$,
its associated ordinal maximin share $\FF_2$-maximin is feasible. Moreover, if an
optimal partition in $\FF_2$ can be computed in polynomial time for every additive valuation $v$, then acceptable allocations can also be computed in polynomial time.}
\end{observation}
\begin{proof}
Given $m$, let $\FF_{m,2}$ be the family of partitions from $\FF_2$ that is associated with $m$.
{let $P^i = (B^i_1, B^i_2) \in \FF_{m,2}$ be an optimal  partition for agent $i$: a partition for which the $\FF_2$-maximin share of agent $i$ with valuation $v_i$ equals to $\min\{v_i(B^i_1), v_i(B^i_2)\}$. }
Let $PS_2=v_2(\items)/2$ be the proportional share for agent $2$.
Since $\max\{v_2(B^1_1), v_2(B^1_2)\}\geq PS_2 \geq \min\{v_2(B^2_1), v_2(B^2_2)\}$, the share is feasible using the partition $(B^1_1, B^1_2)$ in which agent $2$ picks her preferred bundle. {If $P^1$ is computed in polynomial time, then the above allocation algorithm also runs in polynomial time.}
\end{proof}

\subsubsection{Proof of \cref{prop:ordinal}}\label{app:prop-ordinal}

\lemmaOrdinalSelfMaximizing*
\begin{proof}
{For the family $\FF_n$ of partitions and parameter $m$, 
let $\FF_{m,n}$ be the family of partitions associated with the ordinal maximin share $\FF_n$-maximin, a share which we denote by $s$.}
Let $v, v'$ be two arbitrary additive valuation functions. Recall that $\ShareG$
denotes the guarantee of $s$ when the true valuation function is $v$, and the reported valuation function is $v'$. To prove that $\FF_n$-maximin is self maximizing we need to show that $\ShareGvv\geq \ShareG$. 

Let $P = (P_1, \ldots, P_n) \in \FF_{m,n}$ be the partition (of indices) giving $\ShareGvpvp$. That is, when items are indexed in decreasing order of $v'$ value, $P$ induces a partition $(B'_1, \ldots, B'_n)$ of $\items$, and $v'(B'_j) \ge \ShareGvpvp$ for every $j\in [n]$. In contrast, when items are indexed in decreasing order of $v$ value, $P$ induces a possibly different partition of $\items$, and we refer to this partition as $(B_1, \ldots, B_n)$. W.l.o.g., let $B_1$ be such that $v(B_1)$ is smallest. Then $v(B_1) \le \ShareGvv$, because $\ShareGvv$ is determined by the best partition $\FF_{m,n}$, and $P$ is not necessarily the best partition. 

We shall use the following notation. For every $1 \le j \le m$, recall that $e_j$ denotes the $j^{th}$ largest item according to $v$ (here and elsewhere, ties can be broken arbitrarily, without affecting correctness of the proof), similarly, $e'_j$ denotes the $j^{th}$ largest item according to $v'$, $E_j$ denotes the set $\{e_1, \ldots, e_j\}$, and $E'_j$ denotes the set $\{e'_1, \ldots, e'_j\}$.

Let $I = \{i_1, \ldots i_k\}$ be the set of indices in $P_1$, the first part of the partition $P$.  That is, $B_1 = \{e_{i_1}, \ldots, e_{i_k}\}$ and $B'_1 = \{e'_{i_1}, \ldots, e'_{i_k}\}$.  

Let $s_{i_1}$ denote the item of highest index (lowest $v$ value) among the items of $E'_{i_1}$. For every $j$ such that $2 \le j \le k$, let $s_{i_j}$ denote the item of highest index (lowest $v$ value) among the items of $E'_{i_j} \setminus \{s_{i_1}, \ldots, s_{i_{j-1}}\}$. Consider the bundle $S = \{s_{i_1}, \ldots, s_{i_k}\}$.

We claim that $v'(S) \ge v'(B'_1)$ and that $v(S) \le v(B_1)$. Assuming the claim, we infer that $S$ is a bundle {acceptable}  
under $v'$ (recall that $v'(B'_1) \ge \ShareGvpvp ={\ShareGfullvp}$), and that $v(S)$ is not larger than the $\FF_n$-maximin share for $v$ (recall that $v(B_1) \le \ShareGvv$). By reporting $v'$ instead of the true $v$ the agent might get the bundle $S$, and hence her {implied} guarantee does not improve, {as $v(S) \le v(B_1) \le \ShareGvv$}. 
Consequently, {as this holds for any additive $v,v'$,} the ordinal maximin share $\FF_n$-maximin is self-maximizing.

It remains to prove the claim, {showing that $v'(S) \ge v'(B'_1)$ and $v(S) \le v(B_1)$}.

Recall that $S = \{s_{i_1}, \ldots, s_{i_k}\}$ and that $B'_1 = \{e'_{i_1}, \ldots, e'_{i_k}\}$. 
As $v'(s_{i_j})\ge v'(e'_{i_j})$ for every $j$, we get that $v'(S) \ge v'(B'_1)$, as desired.

{It remains to show that $v(S) \le v(B_1)$.} Recall that $B_1 = \{e_{i_1}, \ldots, e_{i_k}\}$, and the notation $E_j$ above. 
{We show that for every $j\in [k]$ the set $S$ can have at most $j-1$ items from $E_{i_j - 1}$.}
{First, observe that $S$ cannot have any item from $E_{i_1 - 1}$, or equivalently, for every $j\in [k]$, the item $s_{i_j}$ does not belongs to $E_{i_1 - 1}$. For $j=1$ this holds because $|E'_{i_1}| = |E_{i_1-1}| + 1$, and hence $E'_{i_1} \setminus E_{i_1 - 1}$ contains at least one item. As every item in $E'_{i_1} \setminus E_{i_1 - 1}$ has larger index (smaller $v$ value)  than every item in $E_{i_1 - 1}$, the item $s_{i_1}$ will be selected from $E'_{i_1} \setminus E_{i_1 - 1}$. For $j$ such that $2 \le j \le k$ the claim that $s_{i_j}\not\in E_{i_1 - 1}$ holds because the set $E'_{i_j}$ contains at least $j$ items not from $E_{i_1 - 1}$ {(because $|E'_{i_j}| \ge i_1 + j - 1$)}, and exactly $j-1$ items of $E'_{i_j}$ are selected in $S$ as $\{s_{i_1}, \ldots, s_{i_{j-1}}\}$. Hence at least one item not from $E_{i_1 - 1}$ remains in $E'_{i_j}$, and can be selected as $s_{i_j}$.}
Hence, the item $e_{i_1}$ of largest $v$ value in $B_1$ has at least as large value of the item of largest $v$ value in $S$. Likewise, $S$ can have at most one item from $E_{i_2 - 1}$, because for every $1 \le j \le k$, the set $E'_{i_j}$ contains at least $j-1$ items not from $E_{i_2 - 1}$. More generally, for every $j\in [k]$ the set $S$ can have at most $j-1$ items from $E_{i_j - 1}$. Hence for every $j$, {as $v$ is additive}, 
$v(\{e_{i_1}, \ldots, e_{i_j}\})$ is at least as large as the value of the $j$ items of highest $v$ value in $S$. As $|S| = |B_1|$, 
this implies that $v(S) \le v(B_1)$, as desired.
\end{proof}

\subsection{Proof of \cref{thm:NS2} }\label{app:thm-NS2}

{We next present the claims that were missing in \cref{sec:additive-n} to complete the proof of \cref{thm:NS2}.}

\begin{lemma}\label{lem:NS-poly-time}
For the class of additive valuations and for every values of $n$ and $1 \le q \le n$, an optimal partition in $\FF_{n,q}$ (that gives the value of the $NS_{n,q}$ share) can be computed in time polynomial in $m$ and $n$.
\end{lemma}

\begin{proof}
Consider a valuation function $v$ over $m$ items. For simplicity of the presentation, we present a polynomial time algorithm for computing an optimal $NS_{n,q}(v)$ partition only in the most difficult case, that of $q=n$. The algorithm can be adapted in a straightforward way to compute optimal $NS_{n,q}(v)$ partitions for values of $q$ that are smaller than $n$ (in the table $T$ described below, for every $i < n - q$, only fill up entries $T_{i,j}$ with $j =i$), and its running time only improves as $q$ becomes smaller. 

Sort the items by their $v$ values. This takes time $O(m \log m)$. We assume for simplicity that the value of  $NS_{n,n}(v)$ is known, and the goal is only to compute a partition that certifies this value. This assumption can be made without loss of generality. For example, observe that $NS_{n,n}(v)$ must be equal to the value of a bundle $B$ that is composed of a set of consecutive items from the prefix $Z$ and a set of consecutive items from the suffix $S$. As there are at most $O(m^4)$ possible such bundles, there are at most $O(m^4)$ candidate values for $NS_{n,n}(v)$. We can run binary search on these candidate values so as the find the highest candidate value for which an  $NS_{n,n}(v)$ partition exists.

Given a candidate value $c$ for $NS_{n,n}(v)$, to check whether there is a partition in $\FF_{n,n}$ of value at least $c$, run the following dynamic programming algorithm. The algorithm maintains a table $T$ with $n$ rows (numbered~1 to $n$) and $m+1$ columns (numbered~0 to $m$), whose entries are integers in the range $[-1,m]$.  Entry $T_{i,j}$ is the smallest legal value value $t$, where $t$
is legal if $t + j \le m$ and the following holds: one may create bins $B_1, \ldots, B_i$ that contain a consecutive partition $(Z_1, \ldots Z_i)$ of the first $j$ items, and a consecutive partition $(S_i, \ldots, S_1)$ 
of the last $t$ items, so that every bin $B_i = Z_i \cup S_i$ has value at least $c$. If there is no such legal value of $t$, then this is denoted by setting $T_{i,j} = -1$. 

The entries of $T$ can be filled from row~1 to row~$n$, and within each row from column~0 to column $m$. Entry $T_{1,j}$ has value $t$ with the smallest legal value of $t$ for which $v(\{e_1, \ldots, e_j\} \cup \{e_{m-t+1}, e_m\}) \ge c$. For $i > 1$, entry $T_{i,j}$ is the smallest legal value of $t$ over all choices of $0 \le j' \le j$ for which $v(\{e_{j'+1}, \ldots, e_j\} \cup \{e_{m-t+1}, \ldots, e_{m - T_{i-1,j'}}\}) \ge c$. The process of filling $T$ takes polynomial time, as there are $O(mn)$ entries and the value of each entry can be computed in time $O(m)$.

Once we fill the table, then we may check whether there is a partition in $\FF_{n,n}$ of value at least $c$ by checking whether row $n$ in the table has an entry that is not $-1$. Finding such an entry $T_{n,j}$, we can also find a corresponding $NS_{n,n}$ partition, by backtracking the steps of the algorithm that gave $T_{nj}$ its value.
\end{proof}

{We shall use $\rho_{n,q}(v)$ to denote the ratio between the value of the $NS_{n,q}$ share and the MMS for additive valuation function $v$, and $\rho_{n,q}$ to denote the infimum ratio over all additive $v$. }

\begin{lemma}
\label{lem:twothirdsapprox}
For the class of additive valuations, for  
{every $n$ and} every $1 \le q \le n$
the value of the $NS_{n,q}$ share is at least a $\frac{2n}{3n-1}$ fraction of the MMS. {In other words, $\rho_{n,q} \ge \frac{2n}{3n-1}$.}
\end{lemma}

\begin{proof}
It suffices to prove the lemma for $q=1$, as then the value of the $NS_{n,q}$ share is smallest. We may also assume that $m > n$, as otherwise the value of the $NS_{n,1}$ share is equal to the MMS.

{The proof proceeds by induction on $n$. For the base case the lemma holds, because $NS_{n,1}$ is equal to the $\MMSi$ when $n=1$. For the inductive step, we assume that the lemma holds for $n-1$, and prove it for $n$.}

For every additive valuation function $v_i$ we show a partition in $\FF_{n,1}$ in which every bundle has value at least $\frac{2n}{3n-1}\MMSi$. The items are ordered from largest to smallest value. We may assume that $v_i(e_1) < \frac{2n}{3n-1}\MMSi$, as otherwise we can set $B_1 = \{e_1\}$, and then use the inductive hypothesis to create an $NS_{n-1,1}$ partition for the remaining items. Every bin in this partition has value at least $\frac{2n}{3n-1}\MMSi$, because the $MMS_{n-1}$ value for the remaining items is at least as large as the $MMS_n$ value for the original items ($n-1$ bundles of the $MMS_n$ partition do not contain $e_1$).

Place items $e_1, \ldots, e_n$ in bins $B_1, \ldots, B_n$, respectively, and continue filling bins in reverse order, closing each bin once its value reaches (or exceeds) $\frac{2n}{3n-1}\MMSi$. We may assume that $B_n$ contains at least three items. It cannot contain only one item because $v_i(e_n) \le v_i(e_1) < \frac{2n}{3n-1}\MMSi$. It can be assumed not to have only two items, because then removing $B_n$ does not decrease the MMS.  (This is a known fact that holds because some bundle in the MMS partition has two of the first $n+1$ items, {and the total value of items $e_n$ and $e_{n+1}$ is at most the value of this pair.}) Consequently, for the remaining items we can use the inductive hypothesis to create an $NS_{n-1,1}$ partition in which every bundle has value at least $\frac{2n}{3n-1}\MMSi$. {(For $q=1$, setting $B_n=\{e_n,e_{n+1}\}$ does not impose any new constraint on $NS_{n-1,1}$ for the remaining items.)}

Hence $v_i(e_n) + v_i(e_{n+1}) < \frac{2n}{3n-1}\MMSi$, implying that $v_i(e_{n+1}) < \frac{n}{3n-1}\MMSi$. Hence no bin can have value as high as $\frac{3n}{3n-1}\MMSi$. Hence the total value of the last $n-1$ bins is at most $(n-1)\frac{3n}{3n-1}\MMSi$, leaving value of at least $(n - (n-1)\frac{3n}{3n-1})\MMSi = \frac{2n}{3n-1}\MMSi$ for $B_1$.
\end{proof}

We state some conventions that are used in all subsequent proofs.
Let $v$ be an arbitrary additive valuation function. We assume that items are sorted from highest to lowest value {(with arbitrary tie breaking)}. 
{When proving lower bounds} we also assume that the MMS is~1, and that in the MMS partition, every bundle has value exactly~1. These assumptions can be made without loss of generality. In particular, if there is a negative example in which the MMS is~1 but the MMS partition contains a bundle of value strictly larger than~1, we can reduce values of items in this bundle until its value becomes exactly~1, and the approximation ratio of $NS_{n,q}$ does not improve. This is because the value of the MMS does not change, and the value of the $NS_{n,q}$ share does not increase, because it is an ordinal maximin share {(and hence monotone in the value of the items)}. {We remark that in explicit examples that we give, the MMS will often not be~1, as we prefer to have integer item values, for readability of the examples.}

{In \cref{lem:rho32} we shall consider $NS_{3,2}$ and prove that for additive valuations $\rho_{3,2} = \frac{4}{5}$. 
Before that, we prove several helpful claims.}
{
\begin{lemma}
\label{lem:rho21}
$\rho_{2,1} = \frac{4}{5}$.
\end{lemma}

\begin{proof}
The example with item values $3,2,2,2,1$ shows that $\rho_{2,1} \le \frac{4}{5}$. Hence it remains to show that $\rho_{2,1} \ge \frac{4}{5}$. 
If $v(e_3) \ge \frac{2}{5}$, then $NS_{2,1}$ may choose $B_2 = \{e_2,e_3\}$ and $B_1$ containing the remaining items, and each bundle has value at least $\frac{4}{5}$. (For $B_2$ this holds because $v(e_2) \ge v(e_3)$, and for $B_1$ this holds because $v(B_2) \le 1$, as otherwise {the MMS partition must contain a bundle containing  two of the top three items, with
value strictly larger than~1.}) If $v(e_3) \le \frac{2}{5}$ then  $NS_{2,1}$ can place $e_1$ in $B_1$, and then items in $B_2$ until the first point in which the value of $B_2$ reaches or exceeds $\frac{4}{5}$. At this point, value of at least $2 - (\frac{4}{5} + \frac{2}{5}) \ge \frac{4}{5}$ is left for $B_1$. In either case, the value of $NS_{2,1}$ is at least $\frac{4}{5}$.
\end{proof}
}

{
\begin{lemma}
\label{lem:rho22}
$\rho_{2,2} = \frac{5}{6}$.
\end{lemma}

\begin{proof}
The example with item values $4,3,2,2,1$ shows that $\rho_{2,2} \le \frac{5}{6}$. Hence it remains to show that $\rho_{2,2} \ge \frac{5}{6}$. 

For every item $j$, we denote $v(e_j)$ by $x_j$, for brevity. If there is an example with $\rho_{2,2} < \frac{5}{6}$ (and with an MMS partition with two bundles of value exactly~1), then it must satisfy the following constraints.

\begin{enumerate}
    \item $x_2 + x_3 \le 1$. Otherwise the MMS partition has a bundle of value more than~1 {(as one of the parts contains two of the three items of highest value)}.
    \item $x_2 + x_3 < \frac{5}{6}$. Otherwise choose $B_2 = \{e_2, e_3\}$ {and we have that $v(B_1)\geq 1$ as $x_2 + x_3 \le 1$.}
    \item $x_4 > \frac{1}{3}$. Otherwise start filling up $B_2$ from $e_2$, until the first point in time in which it reaches or exceeds $\frac{5}{6}$ (and value of at least $\frac{5}{6}$ remains for $B_1$).
    \item $x_2 < \frac{1}{2}$. Otherwise {$x_4\leq x_3 < \frac{1}{3}$}.
    \item $x_1 < \frac{2}{3}$. Otherwise the MMS partition has a bundle of value more than~1 
    {(either $\{e_1\}$ is one of the MMS bundles, and the other bundle has value at least $3x_4 > 1$, or one of the bundles has value at least $x_1+x_4>\frac{2}{3}+\frac{1}{3} = 1$).}
    \item $x_1 + x_2 < \frac{5}{6}$. Otherwise choose $B_1 = \{e_1, e_2\}$, and value of at least $\frac{5}{6}$ remains for $B_2$, {as $x_1 + x_2 < \frac{2}{3}+ \frac{1}{2} =\frac{7}{6}$}.
    \item $x_5 \le \frac{1}{3}$. Otherwise the MMS partition has a bundle of value more than~1 {(one of the MMS bundles has at least three of the five items of highest value, with total value at least $3x_5$).}
\end{enumerate}

Given the above constraints, we start filling up $B_2$ from {$e_3$}, until the first point in time in which it reaches or exceeds $\frac{5}{6}$. At that point $B_2$ must include $e_5$ 
(as $x_3 + x_4 \le x_2 + x_3 < \frac{5}{6}$), and as $x_5 \le \frac{1}{3}$, the value of $B_2$ is at most $\frac{7}{6}$. A value of at least $\frac{5}{6}$ remains for $B_1$, as desired.
\end{proof}
}

{
\begin{lemma}
\label{lem:rho31}
Let $v$ be an additive valuation for which either $v(\{e_2, e_3\}) < \frac{4}{5}\cdot MMS_3(v)$ or $v(\{e_2, e_3\}) > MMS_3(v)$. 
Then $\rho_{3,1}(v) \ge \frac{4}{5}$.
\end{lemma}
}

\begin{proof}
For an additive valuation function $v$ and $n=3$, we write constraints that need to hold if the value of the MMS is~1 (with an MMS partition in which each bundle has value exactly~1, {which is without loss of generality}), {but the value of $NS_{3,1}$ is strictly less than $\frac{4}{5}$.} 
The variables in our constraints are $x_i$ for $i \ge 1$, where $x_i$ denotes the value of item~$e_i$ (namely, $x_i = v(e_i)$). 

\begin{enumerate}
    \item $\sum_i x_i = 3$. {This holds as every bundle in the MMS has value 1, and $n=3$.}
    \item $x_1 < \frac{4}{5}$. Otherwise, take $B_1 = \{e_1\}$, and for the remaining items $NS_{2,2}$ applies, with value $NS_{2,2} \ge NS_{2,1} \ge \frac{4}{5}$. 
    \item $x_3 + x_4 < \frac{4}{5}$. Otherwise, take $B_3 = \{e_3, e_4\}$, and for the remaining items use $NS_{2,1} \ge \frac{4}{5}$. (We use the fact that by removing $\{e_3, e_4\}$ and reducing $n$ to~2, the MMS of the resulting instance remains at least~1 {(because at least two of the top four items are in the same bundle in the original MMS partition, and $e_3$ and $e_4$ have lowest value among these four items)}. 
    
    We infer than $x_4 < \frac{2}{5}$.
    \item $x_5 + x_6 + x_7 \le  1$. Otherwise the MMS partition contains a bundle of value larger than~1. (This follows because some bundle in the MMS partition contains at least three items out of the first seven items, and $e_5, e_6, e_7$ are the items of smallest value among these seven items.)  
    \item $x_3 + x_4 + x_5 \ge \frac{4}{5}$. Otherwise $x_5 \le \frac{4}{15}$. Then, after putting $\{e_3,e_4,e_5\}$ in $B_3$, if we close each of $B_3$ and $B_2$ at the item at which it reaches or exceeds $\frac{4}{5}$,  none of these bundles has value larger than $\frac{16}{15}$ ({recall that $x_2\leq x_1 \le \frac{4}{5}$}), and a value of at least $\frac{13}{15} > \frac{4}{5}$ is left for $B_1$. 
    
    We infer that $B_3 = \{e_3,e_4,e_5\}$
    and $v(B_3) \le \frac{6}{5}$.
       \item $x_2 + x_6 \le 1$.
       Otherwise (namely, if $x_2 + x_6 > 1$), it must be that $e_1$, $e_2$ and $\{e_3, \ldots, e_6\}$  are in separate bundles in the MMS partition (otherwise the MMS partition has a bundle of value strictly larger than~1). Consequently, $B_3$ (shown above to satisfy $B_3 = \{e_3,e_4,e_5\}$) would be  contained in a single bundle of the MMS partition. {As two full bundles of the MMS partition remain, the MMS (with $n=2$) for the remaining items remains at least~1.} Then, from the remaining items, $NS_{3,2}$ can make two bundles of value at least $\frac{4}{5}$, as in the case $NS_{2,1}$. 
    \item $x_2 + x_6 < \frac{4}{5}$.
    Otherwise, make $B_2 = \{e_2, e_6\}$, and as $v(B_2) \le 1$ {(by the constraint $x_2 + x_6 \le 1$)} and we have the constraint $v(B_3) \le \frac{6}{5}$, a value of at least $\frac{4}{5}$ is left for $B_1$. 
    
    We infer that $\{e_2, e_6, e_7\} \subseteq B_2$.

\end{enumerate}

{
If any of the above constraints are violated, then the value of $NS_{3,1}$ is at least $\frac{4}{5}$. Hence we assume that all the above constraints hold. Now we consider two cases, based on the allowed values of {$v(\{e_2, e_3\})= x_2+x_3$.}
}

\begin{enumerate}
    
\item If $x_2 + x_3 \ge 1$, then recall that the last item to enter $B_3$ is $e_5$, and the last item to enter $B_2$ has value not larger than $x_7$. Using $x_5 \le x_4 \le \frac{4}{5} - x_3$ and $x_7 \le x_6 \le \frac{4}{5} - x_2$ we have that the combined value of $B_3$ and $B_2$ is at most $(\frac{4}{5} + \frac{4}{5} - x_3) + (\frac{4}{5} + \frac{4}{5} - x_2) \le \frac{16}{5} - 1 = \frac{11}{5}$.
This leaves for $B_1$ a value of $\frac{4}{5}$.  

\item If $x_2 + x_3 \le \frac{4}{5}$ then $x_3 \le \frac{2}{5}$. If $B_2 = \{e_2, e_6, e_7\}$ then the value of $B_2$ and $B_3$ combined is at most $\frac{11}{5}$ (as $x_5 + x_6 + x_7 \le  1$ and $x_4 \le x_3 \le \frac{2}{5}$), and $\frac{4}{5}$ value remains for $B_1$. If $B_2$ has at least four items, then $B_3$ has value at most $3x_3$ and $B_2$ has value at most $\frac{4}{5} + \frac{1}{2}(\frac{4}{5} - x_2)$ (this is because the two items added to $B_2$ after $e_2$ did not give it a value of $\frac{4}{5}$, implying that the {third item of $B_2$ had value at most $\frac{1}{2}(\frac{4}{5} - x_2)$, and hence so did the item that caused $B_2$ to reach or exceed $\frac{4}{5}$}), for a total of at most $\frac{6}{5} - \frac{1}{2}x_2 + 3x_3 \le \frac{6}{5} + \frac{5}{2}x_3 \le \frac{11}{5}$. Hence again, $\frac{4}{5}$ value remains for $B_1$.

\end{enumerate}
\end{proof}

\begin{lemma}
\label{lem:rho32}
$\rho_{3,2} = \frac{4}{5}$ and thus $\rho_{3,3} \geq \frac{4}{5}$.
\end{lemma}

\begin{proof}
The instance with item values $3,3,2,2,2,2,1$ shows that $\rho_{3,2} \le \frac{4}{5}$.
It remains to show that $\rho_{3,2} \ge \frac{4}{5}$. 

{If it does not hold that $\frac{4}{5} \le v(\{e_2,e_3\} \le 1$, then $\rho_{3,1} \ge \frac{4}{5}$ by \cref{lem:rho31},  and as $\rho_{3,2} \ge \rho_{3,1}$ we conclude that $\rho_{3,2} \ge \frac{4}{5}$.
If, on the other hand,  $\frac{4}{5} \le v(\{e_2,e_3\} \le 1$,
we can form the bundle $B_2=\{e_2, e_3\}$ in the $NS_{3,2}$ partition. As removing any two items of total value at most~1 and decreasing $n$ by~1 does not reduce the MMS (if they were in different bundles in the original MMS partition, then the union of what remains from these bundles can be used as a bundle in the new MMS partition), and as any $NS_{2,1}$ partition for the remaining items gives (together with $B_2$) an $NS_{3,2}$ partition for the original instance, the bound $\rho_{3,2} \ge \frac{4}{5}$ follows from Lemma~\ref{lem:rho21} which shows that $\rho_{2,1} = \frac{4}{5}$.
}
\end{proof}

\subsubsection{Computer assisted analysis of $\rho_{4,3}$}\label{sec:computer}

{In this section we prove the following bound.}

\begin{lemma}\label{lem:rho43-ub}
$\rho_{4,3} = \frac{4}{5}$.
\end{lemma}

The inequality $\rho_{4,3} \le \frac{4}{5}$ is proved by~\cref{example:NS}, {thus we only need to show that $\rho_{4,3} \ge \frac{4}{5}$}.

Our proof {that $\rho_{4,3} \ge \frac{4}{5}$} is based on constructing  {a \em{mixed integer linear program} (MILP)} that searches for {a valuation $v$}
that minimizes the ratio between $NS_{4,3}$ and the MMS (with $n=4$) for $v$. 
As we shall prove {(see \cref{lem:14})}, it suffices to consider instances with up to~14 items.  

The MILP searches for an instance with an MMS partition into four bundles, each of value exactly~5. It attempts to minimize $z$, which represents the value of $NS_{4,3}$. It contains the variables $x_1 \ge \ldots, x_{14} \ge 0$ for the values of the items. 
It also contains binary decision variables $y_j$. 
We write the constraints (and comments explaining them) in the syntax of the online optimizer appspot (https://online-optimizer.appspot.com/), so that readers can simply run the program. {(A text file containing the full program can be obtained from the authors upon request.)}

\begin{verbatim}
    
var x1 >= 0;
...
var x14 >= 0;

var z >= 0;

var y1>=0, binary;
...
var y45>=0, binary;

minimize ns:     z;

subject to

o1: x1 >= x2;
...
o13: x13 >= x14;

\end{verbatim}

The MMS partition imposes the following constraints. Constraint a0 follows from having an MMS partition with four disjoint bundles, each of value~5. 
The other constrain express the fact that for every $k\ge 0$ and $1 \le j \le 3$, there are $j$ bundles in the MMS partition that together contain $jk + j$ of the first $4k + j$ items, {and that the value of these $jk + j$ items (even if these are the least valuable) must be at most $5j$}.

\begin{verbatim}

a0: x1 + x2 + x3 + x4 + x5 + x6 + x7 + x8 + x9 + x10 + x11 + x12 + x13 + x14 = 20; 
a1: x1 <= 5; /* k=0 and j=1 */
a2: x4 + x5 <= 5; /* k=1 and j=1 */
a3: x3 + x4 + x5 + x6 <= 10; /* k=1  and j=2 */
a4: x2 + x3 + x4 + x5 + x6 + x7 <=  15; /* k=1 and j=3 */
a5: x7 + x8 + x9 <= 5; /* k=2 and j=1 */
a6: x5 + x6 + x7 + x8 + x9 + x10 <= 10; /* k=2 and j=2 */
a7: x3 + x4 + x5 + x6 + x7 + x8 + x9 + x10 + x11 <= 15; /* k=2 and j=3 */
a8: x10 + x11 + x12 + x13 <= 5; /* k=3 and j=1 */
a9: x7 + x8 + x9 + x10 + x11 + x12 + x13 + x14 <= 10; /* k=3 and j=2 */

\end{verbatim}

For the next set of constraints, we show that if they do not hold then $z \ge 4$. We shall use the fact that whenever two items of total value not larger than the MMS are removed, the MMS for the remaining items (when $n$ is smaller by one) does not decrease. 

\begin{verbatim}

b1: x1 <= 4; 
b2: x4 + x5 <= 4; 
b3: x6 + x7 + x8 >= 4; 
b4: x7 >= 1; 
b5: x3 + x6 <= 5; 
b6: x2 + x9 <=  5; 

\end{verbatim}

The above constraints are justified as follows:

\begin{itemize}
    \item b1: otherwise choose $B_1 = \{e_1\}$. For the remaining items, take the best $NS_{3,3}$ partition. Then $z \ge 4$, because for the remaining items {$MMS_3(v)\ge 5$}, and $\rho_{3,3} \ge \frac{4}{5}$ {by \cref{lem:rho32}}. 
    \item b2:     {otherwise ({in fact, even if $x4 + x5 \ge 4$ rather than  $x4 + x5 > 4$}), there are two cases to consider. If $x2 + x3 \le 5$ choose $B_2 = \{e_2,e_3\}$ and $B_3 = \{e_4, e_5\}$ (note that $x2 + x3 \ge x4 + x5 \ge 4$), and for the remaining items reduce to $NS_{2,1}$. If $x2 + x3 > 5$, choose $B_4 = \{e_4,e_5\}$, and for the remaining items take the best $NS_{3,1}$ partition (in particular, this will enforce  $e_1\in B_1, e_2\in B_2, e_3\in B_3$). By Lemma~\ref{lem:rho31} (which is applicable because $x2 + x3 > 5$), every bundle in this partition will have value at least~4.}
    
    \item {b3: otherwise, consider an $NS_{4,1}$ partition in which we first put items $e_1, e_2, e_3, e_4$ in bins $B_1, B_2, B_3, B_4$ respectively, and then continue inserting items in decreasing order of value into the bins $B_4$ to $B_1$, where each bin becomes closed (and then we move to the next bin) once its value reaches or exceeds~4. 
    
    In the above $NS_{4,1}$ partition, the last item in $B_4$ is no earlier than $e_6$ {(due to b2, and the fact that its justification also implies that we may assume that $x4 + x5 < 4$)}. Consequently, if b3 does not hold, then the sum of values of last item in bins  
    ${B_2}, B_3, B_4$ is at most~4, and consequently their total value is at most $3\cdot 4 + 4 = 16$. Hence a value of at least~4 remains for $B_1$.}

    We remark that b3 also implies that $x_4 + x_5 + x_6 \ge 4$ (by monotonicity of item values), and together with b2 {(and its justification)}, we deduce that in an $NS_{4,1}$ partition one should take $B_4 = \{e_4, e_5, e_6\}$. 
   
    \item b4: otherwise, choose an $NS_{4,1}$ partition {(as explained in the justification for b3)} with $B_4 = \{e_4, e_5, e_6\}$. {As each of the remaining bins becomes closed once its value reaches or exceeds~4, the value in the bin will be less than~5 (if b4 does not hold. then the last item in each of these bins has value less than~1).} As $v(B_4) \le 6$, each of the remaining bins will indeed reach a value of at least~4.
    
    \item b5: otherwise, none of $e_1, e_2, e_3$ can be in the same bundle as one of $e_4, e_5, e_6$ in the MMS partition, and hence $e_4, e_5, e_6$ are in the same bundle in the MMS partition.  If $4 \le x2 + x3 \le 5$ choose $B_2 = \{e_2,e_3\}$ and $B_3 = \{e_4, e_5, e_6\}$, and for the remaining items reduce to $NS_{2,1}$. Else, choose $B_4 = \{e_4,e_5,e_6\}$, and for the remaining items reduce to $NS_{3,1}$ using Lemma~\ref{lem:rho31}. 
    
    \item b6: otherwise, {none of $e_1, e_2$ can be in the same bundle as one of $e_3, \ldots, e_9$ in the MMS partition, and consequently} $e_3, \ldots, e_9$ fit in two MMS bundles. Choose $B_4 = \{e_4, e_5, e_6\}$ and put $e_3, e_7, e_8, e_9$ in $B_3$. If by this $B_3$ has value at least~4, continue with an $NS_{2,1}$ partition for the remaining items. If this $B_3$ has value at most less than~4, then $x_9 < 1$, and then in the $NS_{4,1}$ partition {presented in the justification for b3}, only $B_4$ might have value above~5 (and no more than~6, by b2). Hence no bin has value less than~4 (because the sum of their values is~20). 
    
\end{itemize}

Each of the above constraints (in groups a and b) captures a condition that must hold independently of other constraints. We next present some other constraints that capture a {logical} \emph{OR} over several basic constraints, {implied by the assumption that there is an MMS partition into four bundles, each of value~5}. For example, one  set of constraints (group d below) captures the constraint that ``From the first {seven} items, either $x2+x7\leq 5$ (when every bundle has at most {two} items), or $x5+x6+x7\leq 5$ (when some bundle has {three} items)".
A basic constraint like ``$x2+x7\leq 5$" is captured by the constraint ``$x2+x7-20*y4\leq 5$", such that if the indicator $y4$ is $0$ then the basic constraint must hold (and if {$y4 = 1$} it might not). To {enforce} {a logical} \emph{OR} over $k$ basic constraints, {we require that} the sum of the $k$ indicators be at most $k-1$, so at least one indicator must be $0$ (this is captured by constraint numbered~0 in each group). 

We next present three groups of constraints (c, d and f), each corresponds to a {logical} \emph{OR}, each regarding the {values for some} prefix of items.

\begin{verbatim}

c1: x2 + x6 - 20*y1 <= 5;
c2: x5 + x6 + x9 - 20*y2 <= 5;
c3: x6 + x7 + x8 + x9 - 20*y3  <= 5;
c0: y1 + y2 + y3 <= 2;

d1: x2 + x7 - 20*y4 <= 5;
d2: x5 + x6 + x7 - 20*y5 <= 5;
d0: y4 + y5 <= 1;

f1: x2 + x8 - 20*y6 <= 5;
f2: x3 + x7 + x8 - 20*y7  <= 5;
f3: x5 + x6 + x7 + x8 - 20*y8  <= 5;
f0: y6 + y7 + y8 <= 2;

\end{verbatim}

The above constraints are justified as follows:

\begin{itemize}
    \item Group c: among the first {nine} items, the MMS partition either has a bundle of size at least {four} (and then c3 holds), or at least two bundles of size {three} (and then c2 holds{: at least out of the first six items must be in the same bundle of size {three}, together with one more item of value at least $x_9$}), or one bundle of size three {and three pairs}, and then either c2 holds (if only one of $e7, e8, e9$ is in a bundle of size {three}) {or c1 holds (if at least two of $e7, e8, e9$ are in a bundle of size three}, then at least one of $e_1, e_2$ is not in a bundle with any of the items $e7, e8, e9$). 
    
    \item Group d: among the first {seven} items, the MMS partition either has a bundle of size at least {three} (and then d2 holds), or has three bundles of size {two} (and then d1 holds). 
    
    \item Group e: among the first {eight} items, the MMS partition either has a bundle of size {four} (and then f3 holds), or has two bundles of size {three} (and then f2 holds), or has at least three bundles of size {at least two} 
    (and then f1 holds). 
\end{itemize}

The above set of constraints can naturally be extended to corresponding constraints when there are $m>14$ items. {This involves the following modifications. 
{For every $i$ such that $14 < i \le m$  we introduce a non-negative variable $x_i$ to denote the value of item $e_i$. } 
For every {$i$ such that} $14 \le i < m$ we add the constraint $o_i: x_i \ge x_{i+1}$. Finally, we replace the constraint $a_0$ by the constraint $\sum_{i=1}^m x_i=20$.}

We are now ready to explain why we can assume that the number of items is at most~14.

\begin{lemma}
\label{lem:14}
{If for $n=4$ there exists an additive valuation $v$ over $m>14$ items such that $NS_{4,3}(v) < 4$ and the constraints above (when extended to $m$ items) hold, then there exists an additive valuation $v'$ over $14$ items such that $NS_{4,3}(v') \le NS_{4,3}(v) < 4$ and the constraints above hold.}
\end{lemma}

\begin{proof}
The constraints of types a, c, d and f are implied by having an MMS partition in which every bundle has value~5, but do not force the instance to have such a partition. 
Suppose now that there is an instance {$I$} with $m \ge 15$ items that satisfies the above constraints {and for which $NS_{4,3} < 4$ (which implies that also all constraints of type b are satisfied)}. Then we claim that we can merge $e_1$ and $e_m$ to a new item {of value $v(\{e_1,e_m\})$}  
(replacing $e_1$), and the new instance {$I'$} will still satisfy all above constraints. {Applying the same argument recursively will result in an instance with $14$ items.}

The only constraints involving $e_1$ are o1, a0, a1 {and b1}.  {They are all assumed to hold in $I$, and we need to prove that they hold in $I'$.}

\begin{itemize}

    \item Constraint o1 holds {in $I'$}, as $x_1$ increased. 

    \item Constraint a0 (if $m > 14$, the interpretation is that the sum 
does not stop at $x_{14}$, but rather includes all items up to $x_m$) 
holds {in $I'$} because the value lost by {$e_m$} 
was added to $e_1$. 

\item As to constraint a1, we consider two cases. If $x_m \ge 1$, {then consider the {\em surplus} of items in $I$, where for every item $e_i$ its surplus is defined as $x_i - 1$ (note that the assumption $x_m \ge 1$ implies that the surplus is non-negative)}.  
If a1 is violated {in $I'$}, then $x_1 + x_m > 5$ {in $I$}, and the total surplus of $\{e_1, e_m\}$ is more than~3. 
Constraint b3 
implies that the total surplus of $\{e_6, e_7, e_8\}$ (and hence also of $\{e_3, e_4, e_5\}$) is at least $4 - 3 = 1$. {Thus the total surplus of items in $I$ is greater than~5, and consequently the sum of item values in $I$ is strictly larger than $15 \cdot 1 + 5 = 20$},  
contradicting {constraint a0 (total value of 20)}. {Hence a1 {holds in $I'$} if $x_m \ge 1$.}
It remains to consider the case that $x_m < 1$. In this case a1 follows because originally $x_1 \le 4$, by constraint b1.  

\item {Constraint b1 {holds in $I'$}, as otherwise in the original instance {$I$} we could form a bundle $B_1 = \{e_1, e_m\}$, and its value would be at least~4 but no more than~5 (as a1 holds {in $I'$}). For the remaining items, we can take the best $NS_{3,3}$ partition, creating bundles $B_2, B_3, B_4$, each of value at least~4 (because for the remaining items $MMS_3(v)\ge 5$, and $\rho_{3,3} \ge \frac{4}{5}$ {by \cref{lem:rho32}}). This contradicts the assumption that $NS_{4,3} < 4$ {for instance $I$}. }
    
\end{itemize}

For every $q$, every $NS_{4,q}$ partition of {$I'$} induces one of the same value for {$I$}, and so if $NS_{4,3} < 4$ for {$I$}, the same holds for {$I'$}.
\end{proof}

The remaining set of constraints (group n) consider various $NS_{4,3}$ partitions, and for each such partition binary decision variables are used to express the requirement that at least one of the bundles has value not larger than $z$ (a logical \emph{OR}). The constraints of type $p_{ij}$ identify various bundles that can be placed in bin $B_i$. The corresponding decision variable is indexed by the index of the constraint, for convenience. Each constraint of type $n_j$ considers one possible $NS_{4,3}$ partition. The index $j$ is a concatenation of the second indices of the $y$ variables in the constraint.

\begin{verbatim}
    
p41: x4 + x5 + x6 - 20*y41 <= z;
p42: x5 + x6 + x7 - 20*y42 <= z;
p43: x6 + x7 + x8 - 20*y43 <= z;
p44: x7 + x8 + x9 - 20*y44 <= z;
p45: x7 + x8 + x9 + x10 + x11 + x12 + x13 - 20*y45 <= z;

p31: x3 + x4 - 20*y31 <= z;
p32: x3 + x7 - 20*y32<= z;
p33: x3 + x7 + x8 - 20*y33 <= z;
p34: x4 + x5 + x9 - 20*y34 <= z;
p35: x4 + x5 + x9 + x10 - 20*y35 <= z;
p36: x4 + x5 + x6 - 20*y36 <= z; 
/* mathematically, y41 can serve instead of y36, but introducing y36 unifies notation */

p21: x2 + x3 - 20*y21 <= z;
p22: x2 + x3 + x10 - 20*y22 <= z;
p23: x2 + x8 - 20*y23 <= z;
p24: x2 + x8 + x9 - 20*y24 <= z;
p25: x2 + x9 - 20*y25 <= z;
p26: x2 + x9 + x10 - 20*y26 <= z;
p27: x2 + x3 + x11 + x12 - 20*y27 <= z;

p11: x1 + x9 + x10 + x11 + x12 + x13 + x14 - 20*y11 <= z;
p12: x1 + x10 + x11 + x12 + x13 + x14 - 20*y12 <= z; 
p13: x1 + x11 + x12 + x13 + x14 - 20*y13 <= z;
/* p14: x1 + x12 + x13 + x14 - 20*y14 <= z; */
p15: x1 + x13 + x14 - 20*y15 <= z;
p16: x1 + x14 - 20*y16 <= z;

n1312: y11 + y23 + y31 + y42 <= 3;
n1321: y11 + y23 + y32 + y41 <= 3;
n2164: y12 + y21 + y36 + y44 <= 3;  
n2412: y12 + y24 + y31 + y42 <= 3; 
n2531: y12 + y25 + y33 + y41 <= 3; 
n3243: y13 + y22 + y34 + y43 <= 3;  
n3631: y13 + y26 + y33 + y41 <= 3;
n5753: y15 + y27 + y35 + y43 <= 3; 
n6165: y16 + y21 + y36 + y45 <= 3; 

    
\end{verbatim}

{Solving this \em{mixed integer linear program} (MILP) gives a valuation $v$ 
that minimizes $z$, where the constraints enforce that $z\leq NS_{4,3}$ and that the MMS is at most $5$.}
{Running the MILP with all the above constraints produces the instance with value $z=4$ and item values $4,4,2,2,2,2,1,1,1,1$. 
{As the minimum value of the MILP is $z=4$, we conclude that $\rho_{4,3} = \frac{4}{5}$ as claimed. We remark that for this instance it holds that}
$NS_{4,3} = NS_{4,2} = 5$, but via an $NS_{4,3}$ partition not considered by any of our constraints of type n {(the constraints of the MILP force the inequality $z\leq NS_{4,3}$, but not necessarily as an equality)}. However, as we know of examples in which the MMS is~5 and $NS_{4,3}=4$, there is no point in strengthening the MILP by additional constraints -- its value will not increase beyond~4.}

\subsection{Tightness}\label{app:NS-tightness}

{We next show that the our analysis of $\rho$ for which the share $NS_{n,q}$ is $\rho$-dominating is tight up to low order terms.}

\begin{proposition}
\label{pro:NEWnegative}
For every $\epsilon > 0$ and fixed $q \ge 1$, there is sufficiently large $n$ and an additive valuation function for which the {$NS_{n,q}$ share} does not provide an agent more than $\frac{2}{3} + \epsilon$ times her MMS. 
\end{proposition}

\begin{proof}
We focus first on the case $q=1$, and later we will address the case $q > 1$.
The valuation function that we describe is taken from~\cite{BEF2021b}, and so is the proof regarding its MMS value.

Let $k =\lceil \frac{1}{2\epsilon}\rceil$, let $n=\sum_{i=0}^k 4^i=\frac{ 4^{k+1} -1}{3}$, and consider a set of $m=3\cdot \sum_{i=0}^k 4^i = 4^{k+1}-1$ items $\items =\big\{(a,b,c) \mid a\in\{0,1,\ldots,k\},b\in\{1,2,3\},c \in \{1,2,\ldots,4^a\}\big\}$. The value of item $(a,b,c)$ is $1-\frac{a}{2k}$ if $b=1$, and $\frac{a}{2k}$ if $b\neq 1$.  
	We call items with $b=1$ large items, and items with $b\in \{2,3\}$ small items.
	For every $0 \le i \le k$, we refer to the items with $a=i$ as belonging to {\em group} $i$. 
	Hence group $i$ contains $4^i$ large items and $2\cdot 4^i$ small items.
	
	{We show that in the partition formed by $NS_{n,1}$, bundle $B_1$ is composed of the three items of group $0$, and hence has value $1+0+0 = 1$, implying that $NS_{n,1} = 1$. 
	In the $NS_{n,1}$ partition, every bundle receives a single large item (as every large items have value at least half, while every small item has value less than half).
	After the allocation of the large items, the $NS_{n,1}$ partition allocates the small items in decreasing order of the groups. Each bundle that received a large item of group $i$ also receives two small items of the same group.
	Thus, bundle $B_1$ will not get another item until all items but the small items of group $0$ are allocated.}

	In contrast,
	the MMS is at least $\frac{3}{2}-\epsilon$. This holds because the items can be partitioned into $n$ bundles, each with value at least $\frac{3}{2}-\epsilon$. 
	For each $a\in \{1,2,\ldots,k\}$, create $2\cdot 4^{a-1}$ bundles by taking two large items of group $a$ and one small item of group $a-1$.
	The value of such a bundle  is $2(1-\frac{a}{2k}) + \frac{a-1}{2k} = 2-\frac{a+1}{2k}\geq \frac{3}{2}-\epsilon  $.
	{In addition, create $ \frac{2\cdot 4^k-2}{3} $ bundles, each containing three small items of group $k$, and each such bundle has a value of $\frac{3}{2}$.} 
	Note that $2\cdot 4^k-2$ is divisible by $3$.
	Finally, create another bundle using  {the two remaining}  small items of group $k$, and the large item of group $0$. {That bundle has value of $2$.} 
	The number of bundles created is $$\sum_{a=1}^k 2 \cdot 4^{a-1} + \frac{2\cdot 4^k-2}{3} +1  =\frac{2\cdot 4^k-2}{3} + \frac{2\cdot 4^k+1}{3} =  n,$$
	which 	concludes the proof.
	
	Now we address the case that $q > 1$. In this case, we require that $n \ge 2q$ (and not just $n \ge q$). As group $k$ contains over half the items, then also for the optimal $NS_{n,q}$ partition, in order that every bin will have value at least~1 it will still be necessary that each of the last $q$ bins contains three items from group $k$. Thereafter, the rest of the bins are filled up exactly as in the case of $NS_{n,1}$, again resulting in $B_1$ having value~1. 
\end{proof}

\subsubsection{Additional Approximation bounds for Nested Shares} \label{app:NS-bounds}

The following proposition presents several examples that provide upper bounds on the approximation ratios obtained by the nested share $NS_{n,q}$, for some values of $n$ and $q$.
It shows that reaching a $\frac{4}{5}$ approximation for $n=3$ requires using $NS_{3,2}$, and $NS_{3,1}$ does not suffice. Moreover, it shows that the bound of $\frac{4}{5}$ does not extend to $NS_{4,2}$. Finally, {we observe that \cref{example:NS} shows that}
even the best possible value of $q$, namely $q=n$, does not offer a better than $\frac{4}{5}$ approximation to the MMS, when $n \ge 3$. 

\begin{proposition}
	\label{pro:NEWnegative3}
	There are additive valuation functions illustrating the following upper bounds on the ratio within which $NS_{n,q}$ approximates the MMS:
	\begin{enumerate}
		\item {$\rho_{2,2} \le \frac{5}{6}$. $\rho_{3,3} \le \frac{5}{6}$.}
		\item $\rho_{3,1} \le 0.77$.
		\item $\rho_{4,2} \le \frac{3}{4}$.
		{More generally, $\rho_{n,n-2} \le \frac{3}{4}$ for all $n \ge 4$.}
		\item $\rho_{n,n} \le \frac{4}{5}$ for all $n \ge 4$. 
	\end{enumerate}
\end{proposition}

\begin{proof}
{For the following additive valuation function the value of $NS_{2,2}$ is 5, whereas the MMS is~6. The item values in non-increasing order are $4,3,2,2,1$. To get a similar example for $NS_{3,3}$, add a single item of value~6.}

For the following additive valuation function the value of $NS_{3,1}$ is 0.7693, whereas the MMS is~1. The item values in non-increasing order are:

0.4615, 0.4615, 0.3846, 0.3846, 0.3846, 0.3077, 0.3077, 0.1539, 0.1539.

The MMS partition is $\{(e_1, e_3, e_8), (e_2, e_4, e_9), (e_5, e_6, e_7)\}$. To get to 0.7693, the $NS_{3,1}$ is forced to create the bundle $(e_3, e_4, e_5)$, and then $(e_2, e_6, e_7)$, but then the items that remain have total value only 0.7693.

For the following additive valuation function the value of $NS_{4,2}$ is~6, whereas the MMS is~8. The item values in decreasing order are:
4, 4, 3, 3, 3, 3, 3, 3, 2, 2, 2.
{We next extend this example to prove that $\rho_{n,n-2} \le \frac{3}{4}$ for all $n \ge 4$.} 
Consider an instance with two items of value~4, with $2(n-1)$ items of value~3, and with $n-1$ items of value~2. The MMS is 8, as the MMS partition can have one bundle of the form $(4,4)$, and $n-1$ bundles of the form $(3,3,2)$. In contrast, $NS_{n,n-2} \le 6$.
In any $NS_{n,n-2}$ partition, each item of value 4 belongs to different bundle. 
{For}  
each of these two bundles to have value more than 6, the four last items must be contained in these bundles. This leaves at most $2(n-1)+(n-1)-4= 3n-7$ items of value 3 and 2. In any partition of these items to $n-2$ bundles, there must be a bundle with at most two items. As each of these items has value at most~3, the value of some bundle is at most 6.

{Finally, an example for which $NS_{4,4} \le \frac{4}{5}$ is  
provided in \cref{example:NS}.
The example extends to every $n>4$ by adding $n-4$ items of value $5$, giving an instance with MMS of value~5, and $NS_{n,n} = 4$ {(as each additional item will be in a different bundle)}. Thus, for $\rho> 4/5$ the share $NS_{n,n}$ does not $\rho$-dominate the MMS.}
{(We note that \cref{example:NS} does not apply to $NS_{3,3}$. The natural way to adapt it to $NS_{3,3}$ is to have item values $(3, 3, 2, 2, 2, 2, 1)$. However, in this case $\{(e_5, e_6, e_7), (e_1, e_4), (e_2, e_3)\}$ is an $NS_{3,3}$ partition with value~5. In this $NS_{3,3}$ partition, bin $B_1$ contains items from the suffix but no items from the prefix.)}

\end{proof}

\end{document}